\documentclass{article}
\usepackage{times}
\usepackage{soul}
\usepackage{url}
\usepackage[hidelinks]{hyperref}
\hypersetup{
	colorlinks=true,
	linkcolor=blue,
	filecolor=green,  
	citecolor=magenta
}
\usepackage[utf8]{inputenc}
\usepackage[T1]{fontenc}
\usepackage[small]{caption}
\usepackage{graphicx}
\usepackage{amsmath}
\usepackage{amsthm}
\usepackage{booktabs}
\usepackage{algorithm}
\usepackage{algorithmic}
\usepackage[switch]{lineno}
\usepackage{makecell}
\usepackage{array,xspace,hhline,tikz,tabularx,amsmath,amsfonts,amsthm,amssymb}
\usepackage{pgflibraryshapes}
\usepackage{pgfplots}
\usepackage{pgfplotstable}
\pgfplotsset{compat=1.7}
\usepgfplotslibrary{fillbetween}
\usepackage{caption}
\usepackage{subcaption}
\usepackage{pgf-pie}
\usepackage{nicefrac}

\usepackage{booktabs}
\usepackage{microtype}
\usepackage[numbers,sort]{natbib}
\usepackage{cleveref}
\usepackage{thm-restate}
\usepackage[shortlabels]{enumitem}

\usepackage{comment}
\usepackage{fullpage}
\allowdisplaybreaks

\newcommand\blfootnote[1]{%
  \begingroup
  \renewcommand\thefootnote{}%
  \NoHyper\footnote{#1}\endNoHyper
  \addtocounter{footnote}{-1}%
  \endgroup
}

\newcounter{remarkcount}
\crefname{remarkcount}{Remark}{Remarks}
\crefname{remark}{Remark}{Remarks}

\newtheorem{definition}{Definition}

\newtheorem{theorem}{Theorem}

\newcommand{\nash}{$\mathsf{NASH}$\xspace}
\newcommand{\cut}{$\mathsf{CUT}$\xspace}
\newcommand{\util}{$\mathsf{UTIL}$\xspace}
\newcommand{\egal}{$\mathsf{EGAL}$\xspace}
\newcommand{\fut}{$\mathsf{FUT}$\xspace}
\newcommand{\map}{$\mathsf{MP}$\xspace}

\begin{document}

	\date{}
	\title{Approximate Strategyproofness in Approval-based Budget Division}

\author{
	Haris Aziz$^\ast$
	\quad
	Patrick Lederer$^\dagger$
	\quad
	Jeremy Vollen$^\ddagger$
}

	\maketitle
\blfootnote{$^\ast$UNSW Sydney. Email: \texttt{haris.aziz@unsw.edu.au}}
\blfootnote{$^\dagger$ILLC, University of Amsterdam. Email: \texttt{p.lederer@uva.nl}}
\blfootnote{$^\ddagger$Northwestern University. Email: \texttt{vollen@northwestern.edu}}

	\begin{abstract}
		In approval-based budget division, the task is to allocate a divisible resource to the candidates based on the voters' approval preferences over the candidates. For this setting, \citet{BBPS21a} have shown that no distribution rule can be strategyproof, efficient, and fair at the same time. In this paper, we aim to circumvent this impossibility theorem by focusing on approximate strategyproofness. To this end, we analyze the incentive ratio of distribution rules, which quantifies the maximum multiplicative utility gain of a voter by manipulating. While it turns out that several classical rules have a large incentive ratio, we prove that the Nash product rule (\nash) has an incentive ratio of $2$, thereby demonstrating that we can bypass the impossibility of \citeauthor{BBPS21a} by relaxing strategyproofness. Moreover, we show that an incentive ratio of $2$ is optimal subject to some of the fairness and efficiency properties of \nash, and that the positive result for the Nash product rule even holds when voters may report arbitrary concave utility functions. Finally, we complement our results with an experimental~analysis.
	\end{abstract}

	\section{Introduction}
	
	In many applications of collective decision-making, the goal is to distribute a divisible resource to some candidates based on the voters' preferences over the candidates. For instance, this task arises in (fractional) participatory budgeting \citep{AzSh21a,RSM25a}, where a budget needs to be distributed to projects based on the voters' preferences over the projects, and donor coordination \citep{BBP+19a,BGSS23a}, where monetary donations are allocated to charities based on the donors' preferences. In this paper, we will study such budget division problems based on approval preferences, i.e., voters only distinguish between liked and disliked candidates. This model, which we refer to as \emph{approval-based budget division}, has been introduced by \citet{BMS05a}. Moreover, this setting has attained significant attention \citep[e.g.,][]{ABM20a,ALL+25b,MPS20a,BBPS21a,KrPe25b} because approval ballots offer a favorable tradeoff between the cognitive burden for voters and the expressive power of the ballot format \citep{FBG23a,YHPP24a}. Formally, we will thus study \emph{distribution rules}, which map each approval profile to a distribution of the resource to the candidates.
	
	The literature on such distribution rules \citep[e.g.,][]{ABM20a,BBPS21a} has identified three fundamental properties known as \emph{strategyproofness}, \emph{efficiency}, and \emph{fairness}. Roughly, strategyproofness requires that voters cannot benefit by lying about their true preferences. Or, put differently, this property ensures that it is in the best interest of the voters to reveal their preferences truthfully. 
	Efficiency, on the other hand, states that we need to choose a distribution that is on the Pareto frontier, i.e., there is no other distribution that makes all voters weakly better off and some voters strictly better off. Lastly, fairness demands that the chosen distribution should be fair to all groups of voters: an appropriate amount of the resource should be spent on behalf of each group of voters. This idea has led to numerous fairness notions such as average fair share or the core, which postulate lower bounds on the utilities of (groups of) voters \citep[see, e.g.,][]{ABM20a}.
	
	Unfortunately, \citet{BBPS21a} have shown in a sweeping impossibility theorem that these three desiderata cannot be jointly satisfied: every strategyproof and efficient distribution rule violates even minimal fairness conditions. Moreover, while there are reasonable strategyproof distribution rules, such as the utilitarian rule and the conditional utilitarian rule, these rules severely violate efficiency or fairness conditions \citep{ABM20a,ALL+25b}. On the other hand, nothing is known about how severely fair and efficient rules fail strategyproofness, thereby leaving it open whether we can escape the impossibility theorem of \citeauthor{BBPS21a} by relaxing strategyproofness. 
	
	\paragraph{Contribution.} In this paper, we investigate whether there are distribution rules that are efficient, fair, and approximately strategyproof. To this end, we introduce the \emph{incentive ratio} of a distribution rule as the worst-case ratio between the utility of a voter when manipulating and when voting honestly. Put differently, if a rule has an incentive ratio of $\gamma$, a voter can get at most $\gamma$ times her original utility by manipulating. Hence, the lower the incentive ratio, the less manipulable a rule is, and an incentive ratio of $1$ means that the rule is strategyproof. The incentive ratio has previously been studied in private good settings, such as Fisher markets~\citep[e.g.,][]{CDZ11a,CDZZ12a,CDT+22a} and assignment problems \citep[e.g.,][]{HWWZ24a,BTWY25a}, where it led to several positive results: desirable, manipulable mechanisms often have a low incentive ratio. In practice, this may suffice to prohibit manipulations as the cost of a manipulation (e.g., to learn the other voters' preferences or to compute a beneficial manipulation) may exceed its benefit.
	
	Unfortunately, it turns out that many distribution rules have a large incentive ratio. Specifically, we prove that the incentive ratios of the fair utilitarian rule (\fut) by \citet{BMS02a}, the maximum payment rule (\map) by \citet{ALL+25b}, and the egalitarian rule (\egal) grow linearly in the numbers of voters or candidates, thereby demonstrating that these rules are severely manipulable. 
	By contrast, we show that the Nash product rule (\nash) has an incentive ratio of~$2$. Since this rule satisfies efficiency and demanding fairness axioms, this result answers our question about the existence of fair, efficient, and approximately strategyproof distribution rules in the positive. Moreover, we show that \nash achieves its incentive ratio of $2$ even when voters are allowed to report concave utility functions over the distributions. Thus, by relaxing strategyproofness, we circumvent the impossibility theorem of \citet{BBPS21a} for a much richer class of utility functions than approval utilities.

	Motivated by this positive result, we further explore whether there are distribution rules that have an incentive ratio smaller than $2$ and satisfy some of the desirable properties of \nash. Perhaps surprisingly, it turns out that, under various side conditions, the incentive ratio of \nash is optimal. In more detail, we show for the following classes that all corresponding rules have an incentive ratio of at least $2$: \emph{(i)} distribution rules that satisfies average fair share, one of the fairness conditions of \nash; \emph{(ii)} distribution rules that maximize an additively separable and strictly concave welfare function; and \emph{(iii)} distribution rules that satisfy efficiency, a mild fairness axiom called group fair share, and three basic side conditions. These results demonstrate that, subject to its desirable properties, no distribution rule has a lower incentive ratio than \nash. 
	
	Lastly, we experimentally analyze the manipulability of distribution rules. Specifically, we
	measure the average incentive ratio and the percentage of manipulable profiles of \fut, \map, \egal, and \nash for approval profiles sampled from several distributions. These experiments show that all four rules are frequently manipulable, with \nash and \egal allowing for manipulations in almost all considered profiles. Further, with the exception of \fut, our theoretical results on the incentive ratio match the trends observed in the experiments. That is, while \nash is effectively always manipulable, voters never gain substantially by lying about their preferences, whereas \egal and \map allow for much more beneficial manipulations. The outlier to this trend is \fut, whose average incentive ratio is much lower than our theoretical worst-case result indicates.
	
	\paragraph{Related Work.}\label{subsec:relatedwork}
	The work on approval-based budget division was initiated by \citet{BMS02a,BMS05a} and we refer to the survey by \citet{SuTe26a} for an overview of this topic as well as related models. More specifically, our paper directly builds on the work of \citet{BBPS21a}, who have shown that every strategyproof and efficient distribution rule violates even minimal fairness conditions. This impossibility theorem strengthens prior results by \citet{BMS05a} and \citet{Dudd15a}. In light of these negative results, researchers have also examined how well distribution rules approximate efficiency and fairness properties. Specifically, \citet{ABM20a} have proven that two strategyproof rules are severely inefficient, whereas \citet{ALL+25b} have shown that most rules severely fail a fairness notion called average fair share. We refer to \Cref{fig:rulesoverview} for an overview of these results. 
	
	\begin{figure}
		\centering
		\renewcommand{\arraystretch}{0.75}
		\setlength{\tabcolsep}{0.5em}	
		\setcellgapes{1.75pt}   
		\makegapedcells
		\begin{tabular}{cccc}
			\toprule
			Rule & Strategyproofness & Efficiency & Fairness (AFS)\\
			\midrule
			\util & $1^\dagger$ & $1^\dagger$ & \makecell{$\infty^\dagger$}\\
			\cut & $1^\dagger$ & $\mathcal{O}(\sqrt[3]{n})^\dagger$ & \makecell{$\Theta(n)^\ddagger$\\(GFS)$^\dagger$}\\
			\nash & \textcolor{blue!80!black}{${2}$} & $1^\dagger$ & \makecell{$1^\dagger$} \\
			\egal & \makecell{\textcolor{blue!80!black}{$\Theta(m)$}\\(Excl. strategypr.)$^\dagger$} & $1^\dagger$ & \makecell{$\Theta(m)^\dagger$\\(IFS)$^\dagger$}\\
			\fut & \makecell{\textcolor{blue!80!black}{$\Theta(n)$}\\ (Monotonicity)$^\ast$} & $1^\ast$ & \makecell{$\Theta(n)^\ddagger$\\ (GFS)$^\ast$}\\
			\map & \makecell{\textcolor{blue!80!black}{$\Theta(n)$}\\ (Monotonicity)$^\ddagger$} & $\Theta(\log n)^\ddagger$ & \makecell{$2^\ddagger$\\ (GFS)$^\ddagger$}\\
			\bottomrule
		\end{tabular}
		\caption{Approximation ratios of common distribution rules to strategyproofness, efficiency, and fairness. In each row, we show for the given rule its approximation ratio to strategyproofness (i.e., its incentive ratio), efficiency, and fairness. Approximate efficiency and fairness are measured by the minimal values $\alpha$ and $\beta$ such that, when multiplying the utilities of all voters under the given rule by $\alpha$ (resp. $\beta$), efficiency (resp. AFS) is guaranteed to be satisfied. 
			If a rule satisfies an axiomatic weakening of strategyproofness or AFS, this is indicated in brackets. 
			Our results are marked in blue. Results marked by~a dagger ($\dagger$), asterisk ($\ast$), and double dagger ($\ddagger$) can be found in \citet{ABM20a}, \citet{BBPS21a}, and \citet{ALL+25b}, respectively. These references also provide the definitions of all rules and axioms in the table, most of which are recapped in  \Cref{sec:preliminaries}.
		}
		\label{fig:rulesoverview}
	\end{figure}
	
	We note that our results give a strong argument in favor of \nash as its incentive ratio is optimal subject to its fairness guarantees. Prior work has already demonstrated that this rule is appealing from other perspectives. For example, \citet{ABM20a} showed that \nash satisfies two demanding fairness notions called average fair share and the core, and \citet{BBP+19a} proved that it gives strong participation incentives. Further, \citet{GuNe14a} axiomatically characterized \nash based on a mild fairness notion, efficiency, and several auxiliary conditions. It is also known that, among fair rules, \nash gives almost optimal guarantees regarding the utilitarian social welfare \citep{MPS20a}.
	
	Finally, the incentive ratio has been studied before. Specifically, it was first introduced for Fisher markets, where goods are sold to agents \citep{CDZ11a,CDZZ12a,CDTZ16a,CDT+22a}. For example, for linear utilities, the incentive ratio of these markets is $2$.
	Moreover, the incentive ratio has also been analyzed for rank aggregation \citep{EbLe26d} and assignment mechanisms, such as the probabilistic serial rule \citep{HWWZ24a, LSX24a}, picking sequences \citep{XiLi20a}, and cake cutting \citep{BTWY25a}.
	\section{Preliminaries}\label{sec:preliminaries}
	
	For each integer $t\in \mathbb{N}$, we define $[t]=\{1,\dots, t\}$. We assume there is a set of $n$ voters $N=[n]$ and a set of $m$ candidates $C=\{x_1,\dots,x_m\}$. All voters have \emph{dichotomous preferences} over the candidates, i.e., they only distinguish between liked and disliked candidates. To indicate her preferences, each voter $i\in N$ reports an \emph{approval ballot} $A_i\subseteq C$, which is the non-empty set of her approved candidates. Further, an \emph{approval profile} $A=(A_i)_{i\in N}$ is a vector specifying the approval ballot of each voter $i\in N$. Given an approval profile $A$, we denote by $A_{-i}=(A_1,\dots, A_{i-1}, A_{i+1},\dots, A_n)$ the profile derived from $A$ by removing voter $i$ and by $(A_{-i},A_i')=(A_1,\dots, A_{i-1}, A_i', A_{i+1},\dots, A_n)$ the profile where we change the approval ballot of voter $i$ to $A_i'$. 
	
	Given an approval profile, our goal is to distribute a continuous resource to the candidates. We will formalize this with the help of \emph{distributions}, which are functions $p:C\rightarrow [0,1]$ such that $\sum_{x\in C} p(x)=1$. Further, the set of all distributions over the set of candidates $C$ is denoted by $\Delta(C)$. Less formally, a distribution $p$ specifies for every candidate $x\in C$ the share $p(x)$ of the resource that is allocated to $x$. To select a distribution for an input profile, we will use \emph{(distribution) rules}, which map each approval profile $A$ (over any finite sets of voters $N$ and candidates $C$) to a distribution $p\in \Delta(C)$. We denote by $f(A,x)$ the share of the resource that a rule $f$ assigns to a candidate $x$ in a profile~$A$. 
	
	Finally, we assume that the utility of every voter $i\in N$ for a distribution $p$ is the total share that $p$ assigns to the approved candidates of the voter, i.e., $u_i(p)=\sum_{x\in A_i} p(x)$.
	
	\subsection{Fairness and Efficiency}\label{subsec:eff+fair}
	
	Since we aim to find fair and efficient rules that are approximately strategyproof, we next introduce the standard notion of (Pareto) efficiency and three fairness axioms. We note that we only define these conditions for distributions; a distribution rule satisfies a given axiom if, for every profile, its chosen distribution satisfies the axiom. 
	
	\paragraph{Efficiency.} Efficiency requires that it is impossible to make some voter better off without making any other voter worse off. Formally, a distribution $p$ is \emph{(Pareto) efficient} for a profile $A$ if there is no other distribution $q$ such that $u_i(q)\geq u_i(p)$ for all $i\in N$ and $u_i(q)>u_i(p)$ for some~$i\in N$.
	
	\paragraph{Positive share.} Positive share is a minimal fairness axiom demanding that each voter gets non-zero utility. Hence, a distribution $p$ satisfies \emph{positive share (PS)} for a profile $A$ if $u_i(p)>0$ for all $i\in N$. 
	
	\paragraph{Group fair share.} The idea of group fair share is that, for each group of voters $S$, we should spend a share proportional to the size of $S$ on the approved candidates of the group. More formally, a distribution $p$ satisfies \emph{group fair share (GFS)} for a profile $A$ if, for all groups of voters $S\subseteq N$, it holds that $\sum_{x\in \bigcup_{i\in S} A_i} p(x)\geq \frac{|S|}{n}$. GFS is a fairly mild fairness condition which is satisfied by many known rules. 
	
	\paragraph{Average fair share.} Average fair share requires that, if all voters in a group $S$ approve a common candidate, then the average utility of the group $S$ should be at least proportional to the size of $S$. Put differently, a distribution $p$ satisfies \emph{average fair share (AFS)} for a profile $A$ if $\frac{1}{|S|} \sum_{i\in S} u_i(p)\geq \frac{|S|}{n}$ for all groups of voters $S\subseteq N$ with $\bigcap_{i\in S} A_i\neq \emptyset$. Among all commonly studied rules, only the Nash product rule satisfies~AFS. 
	
	\subsection{Strategyproofness and Incentive Ratio}\label{subsec:IR}
	
	We now turn to the central properties of this paper, namely strategyproofness and the incentive ratio of distribution rules. First, strategyproofness requires that voters cannot increase their utility by misreporting their true approval ballots. More formally, a distribution rule $f$ is \emph{strategyproof} if $u_i(f(A))\geq u_i(f(A_{-i}, A_i'))$ for all voters $i\in N$, profiles $A$, and approval ballots~$A_i'$. Conversely, a rule is said to be \emph{manipulable} if it is not strategyproof. Strategyproofness is desirable because, for manipulable rules, we cannot expect voters to report their true preferences, which may lead to suboptimal outcomes. 
	
	Unfortunately, strategyproofness is a rather prohibitive axiom for approval-based budget division. Specifically, \citet{BBPS21a} have shown that every strategyproof and efficient rule fails even the minimal condition of positive share. 
	
	\begin{theorem}[\citeauthor{BBPS21a}, \citeyear{BBPS21a}]\label{thm:BBPS}
		No strategyproof distribution rule satisfies both efficiency and positive~share.
	\end{theorem}
	
	To circumvent this impossibility theorem, we will consider an approximate variant of strategyproofness. In more detail, we will analyze the incentive ratio of distribution rules, which quantifies the worst-case multiplicative utility gain of a voter.
	
	\begin{definition}[Incentive Ratio]
		The \emph{incentive ratio} of a distribution rule $f$ is defined by $\gamma(f)=\sup_{A, i, A_i'} \frac{u_i(f(A_{-i}, A_i'))}{u_i(f(A))}$.
	\end{definition}
	
	Put differently, if a distribution rule has an incentive ratio of $\gamma$, a voter can get at most $\gamma$ times her original utility by manipulating. Since the utility of a voter can be $0$ in the truthful profile $A$, we use the convention that $\frac{0}{0}=1$ and $\frac{x}{0}=\infty$ for all $x>0$. Further, we note that a voting rule has an incentive ratio of $1$ if and only if it is strategyproof, and a smaller incentive ratio means that a rule is less manipulable. We refer to \Cref{fig:example} for an example of the incentive ratio and an illustration of the rules defined in the next section.
	
	\subsection{Distribution Rules}\label{subsec:rules}
	
	We next introduce the four distribution rules that we will analyze in this paper. We note that all of the following rules fail strategyproofness but are efficient (except for \map) and satisfy varying degrees of fairness: \egal satisfies PS, \fut and \map satisfy GFS, and \nash even satisfies AFS.
	
	\paragraph{Nash product rule.} The Nash product rule (\nash) returns a distribution $p$ that maximizes the Nash welfare $\prod_{i\in N} u_i(p)$. While there can be multiple such distributions, all of them induce the same utility vector $(u_i(p))_{i\in N}$, so we do not need to worry about the tie-breaking for our purposes.
	
	\paragraph{Egalitarian rule.} The egalitarian rule (\egal) lexicographically maximizes the minimal utility of a voter: we first maximize the utility of the worst-off voter, subject to that we maximize the utility of the second worst-off voter, and so on. 
	To make this more formal, let $u^{\min}(p)$ denote a vector that contains the utilities $u_i(p)$ of all voters $i\in N$ in increasing order. 
	Then, \egal chooses a distribution $p$ that lexicographically maximizes $u^{\min}$, i.e., for every other distribution~$q$, it either holds that $u^{\min}(p)=u^{\min}(q)$ or there is an index $\ell\in [n]$ such that $u^{\min}(p)_i=u^{\min}(q)_i$ for all $i\in [\ell-1]$ and $u^{\min}(p)_\ell>u^{\min}(q)_\ell$. 
	Just as for \nash, there can be multiple distributions that lexicographically optimize $u^{\min}$, but all of them induce the same utility vector $(u_i(p))_{i\in N}$.
	
	\paragraph{Maximum payment rule.} The maximum payment rule (\map) was introduced by \citet{SpGe19a} and recently studied for budget division by \citet{ALL+25b}. It computes the outcome distribution greedily and first gives to each voter control over a $\frac{1}{n}$ share of the budget. Then, it repeatedly identifies the most approved candidate $x$ (with ties broken lexicographically), assigns the shares of the voters that approve $x$ to this candidate, and deletes $x$ as well as all voters that approve $x$ from the profile. This process is repeated until each voter has allocated her $\frac{1}{n}$ share and the final distribution corresponds to the total shares accumulated by each candidate.
	
	\paragraph{Fair utilitarian rule.} The fair utilitarian rule (\fut) was first suggested by \citet{BMS02a} and later rediscovered by \citet{BBPS21a}. In this rule, each voter $i\in N$ controls a share of $\frac{1}{n}$ and is assigned a weight $\lambda_i$, which is initially equal to $1$. Moreover, we let $t=\max_{x\in C} |\{i\in N\colon x\in A_i\}|$ denote the maximum approval score of a candidate. Then, \fut repeats the following routine until all voters have spent their share: we first increase the weights $\lambda_i$ of all voters who have not spent their share yet until the set of candidates $X$ with $\sum_{i\in N:x\in A_i} \lambda_i=t$ is non-empty. Note that in the first round, we do not need to increase the weights for this. Next, each voter $i$ who has not spent her share and approves a candidate in $X$ allocates her $\frac{1}{n}$ share uniformly to the candidates in $X\cap A_i$. Lastly, we delete all candidates in $X$ from the profile and repeat this process if some voter has not spent her share yet.
	
	\section{Theoretical Results}
	
	We are now ready to present our theoretical results. Specifically, we prove in \Cref{subsec:notNash} that \map, \fut, and \egal have large incentive ratios. Furthermore, in \Cref{subsec:Nash}, we show that \nash has an incentive ratio of $2$ and in \Cref{subsec:Nashoptimality} that this is optimal subject to the desirable properties of \nash. 
	
	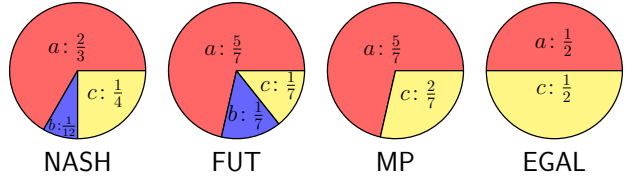
\begin{figure}
		\centering
		\scalebox{0.6}{
			\begin{tikzpicture}
				\pie[radius=1.5, sum=12, text=inside, hide number, color={red!60,blue!60,yellow!60}, pos={0,0}]{8/{\Large$a\!:\frac{2}{3}$}, 1/\phantom{\LARGE ${{{I^I}^I}^I}^I\!\!\!\!\!\!\!\!\!\!\!\!\!$}$b\!:\!\!\frac{1}{12}$, 3/\Large$c\!:\frac{1}{4}$};
				\node at (0, -2) {\LARGE\nash};
			\end{tikzpicture}
			\hspace{0.3cm}
			\begin{tikzpicture}
				\pie[radius=1.5, sum=7, text=inside, hide number, color={red!60,blue!60,yellow!60}]{5/\Large$a\!:\frac{5}{7}$, 1/\phantom{\LARGE ${I^I}^I\!\!\!\!\!\!\!\!$}\Large$b\!:\frac{1}{7}$, 1/\Large$c\!:\frac{1}{7}$};
				\node at (0, -2) {\LARGE\fut};
			\end{tikzpicture}	
			\hspace{0.3cm}
			\begin{tikzpicture}
				\pie[radius=1.5, sum=7,  text=inside, hide number, color={red!60,yellow!60}]{5/\Large$a\!:\frac{5}{7}$, 2/\Large$c\!:\frac{2}{7}$};
				\node at (0, -2) {\LARGE\map};
			\end{tikzpicture}	
			\hspace{0.3cm}
			\begin{tikzpicture}
				\pie[radius=1.5, sum=2,  text=inside, hide number, color={red!60,yellow!60}]{1/\Large$a\!:\frac{1}{2}$, 1/\Large$c\!:\frac{1}{2}$};
				\node at (0, -2) {\LARGE\egal};
			\end{tikzpicture}	
		}
		\caption{Example of our distribution rules and the incentive ratio. The pie charts display the distributions chosen by \nash, \fut, \map, and \egal for the profile $A=(2\!: \{a\},  3\!: \{a,b\}, 1\!: \{b,c\}, 1\!: \{c\})$. To illustrate the incentive ratio, let $A'=(2\!: \{a\},  3\!: \{a,b\}, 1\!: \{b\}, 1\!: \{c\})$ be the profile where the sixth voter deviated to $\{b\}$. For this profile, \egal chooses the distribution $q$ with $q(a)=q(b)=q(c)=\frac{1}{3}$. Thus, the utility of the manipulator under \egal increases from $\frac{1}{2}$ to $\frac{2}{3}$, which yields an incentive ratio of $\frac{2/3}{1/2}=\frac{4}{3}$.
		}
		\label{fig:example}
	\end{figure}
	
	\subsection{Distribution Rules with Large Incentive Ratios}\label{subsec:notNash}
	
	As our first contribution, we prove that the incentive ratios of \map, \fut, and \egal grow linearly in the number of voters and candidates, respectively. This means that, in the worst-case, these rules can be severely manipulated. Moreover, we note that all three rules satisfy axiomatic weakenings of strategyproofness: \map and \fut are monotonic \citep{BBPS21a,ALL+25b}, whereas \egal satisfies a notion called excludable strategyproofness \citep{ABM20a}. Our result thus also demonstrates that these properties do not give convincing guarantees regarding strategyproofness.

	\begin{restatable}{theorem}{nonash}\label{thm:nonash}
		The following claims hold:
		\begin{enumerate}[leftmargin=*, label=(\arabic*)]
			\item The incentive ratio of \map is in $\Theta(n)$. 
			\item The incentive ratio of \fut is in $\Theta(n)$. 
			\item The incentive ratio of \egal is in $\Theta(m)$. 
		\end{enumerate}
	\end{restatable}
	\begin{proof}[Proof of Claim (1)] We only prove Claim~(1) here and defer the proofs of the other statements to \Cref{app:notNash}. Now, for the upper bound on the incentive ratio of \map, we note that every voter spends a share of $\frac{1}{n}$ on her approved candidates, so the distribution $p$ chosen by this rule satisfies that $u_i(p)\geq\frac{1}{n}$ for all voters $i\in N$. Since the utility of a voter is at most $1$, this implies that the incentive ratio of \map is upper bounded by $n$. 
		
		For the lower bound, we fix an integer $k\geq 2$ and consider the following profiles with $m=3$ candidates and $n=2k+4$ voters, which only differ in the ballot of the right-most voter. Further, we suppose the tie-breaking of \map favors $a$ to $c$ to~$b$.\smallskip
		
		{\tabcolsep=0.5em
			\noindent\begin{tabular}{lllllll}
				$A$: & $k$: $\{b,c\}$ & $3$: $\{a,b\}$ & $k-1$: $\{a\}$ & $2$: $\{c\}$ \\
				$A'$: & $k$: $\{b,c\}$ & $3$: $\{a,b\}$ & $k-1$: $\{a\}$ & $1$: $\{c\}$ & $1$: $\{a,c\}$\smallskip
			\end{tabular}		
		}
		
		For $A$, \map chooses the distribution $p$ with $p(b)=\frac{k+3}{2k+4}$, $p(a)=\frac{k-1}{2k+4}$, and $p(c)=\frac{2}{2k+4}$. Specifically, \map first assigns a share of $\frac{k+3}{2k+4}$ to $b$ as it is approved by the most voters. After deleting $b$ and the corresponding voters, the remaining voters report singleton ballots. By the definition of \map, it thus follows that $a$ obtains a share of $\frac{k-1}{2k+4}$ and $c$ of $\frac{2}{2k+4}$. 
		
		By contrast, both $a$ and $b$ are approved by $k+3$ voters in $A'$. By our tie-breaking assumption, this means that \map first assigns a share of $\frac{k+3}{2k+4}$ to $a$. Next, after deleting $a$ and the corresponding voters, all remaining voters approve $c$, so \map sends the shares of these voter to $c$. Thus, \map selects the distribution $q$ with $q(a)=\frac{k+3}{2k+4}$ and $q(c)=\frac{k+1}{2k+4}$ for $A'$. 
		
		Finally, let $i$ be the voter who deviates from $\{c\}$ to $\{a,c\}$ in our profiles. Assuming that $\{c\}$ is true ballot of $i$, the utility of this voter is $u_i(p)=\frac{2}{2k+4}$ in $A$ and $u_i(q)=\frac{k+1}{2k+4}$ in~$A'$. This means for the incentive ratio of \map that
		$\gamma(\text{\map})\geq \frac{u_i(q)}{u_i(p)}=\frac{k+1}{2}=\frac{n-3}{4}$. Hence, we conclude that $\gamma(\text{\map})\in \Theta(n)$. 
	\end{proof}
	
	\subsection{The Incentive Ratio of \nash}\label{subsec:Nash}
	
	We now turn to \nash and show that its incentive ratio is $2$. Since \nash is efficient and satisfies AFS, this result answers our original research question in the affirmative: there are efficient, fair, and approximately strategyproof distribution rules, so we can escape \Cref{thm:BBPS} by relaxing strategyproofness.
	
	Even more, we will show that \nash has an incentive ratio of $2$ for a much richer class of utility functions than approval preferences. To formalize this, we assume in this section that each voter $i$ reports a \emph{utility function} $u_i:\Delta(C)\rightarrow \mathbb{R}_{\geq 0}$ over the distributions instead of an approval ballot. Further, we require utility functions to be \emph{(i)} concave (i.e., $u_i(\lambda p + (1-\lambda) q)\geq \lambda u_i(p)+(1-\lambda) u_i(q)$ for all distributions $p,q\in \Delta(C)$ and $\lambda\in [0,1]$) and \emph{(ii)} non-zero (i.e., there is a distribution $p\in \Delta(C)$ with $u_i(p)>0$). Since the non-zero condition is almost trivial, we will refer to such utilities simply as \emph{concave}. We note this setup is strictly more general than the approval setting since each approval ballot $A_i$ can be treated as the (concave) utility function $u_i(p)=\sum_{x\in A_i} p(x)$. Moreover, our assumptions on utility functions are very mild and include numerous other models, such as linear utilities, Leontief utilities, and Cobb-Douglas utilities \citep[see, e.g.,][Ch.~4]{Vari99a}. Further, a (concave) \emph{utility profile} $U=(u_i)_{i\in N}$ contains the (concave) utility function of every voter $i\in N$ and we adopt notation analogous to approval profiles. In particular, $(u_i', U_{-i})$ is the utility profile derived from $U$ by changing the utility function of voter $i$ to $u_i'$. 
	
	We note that the definition of the incentive ratio immediately generalizes to our enriched setting: the incentive ratio of a distribution rule $f$ is $\gamma(f)=\sup_{U,i, u_i'} \frac{u_i(f(u_i', U_{-i}))}{u_i(f(U))}$, where the supremum is taken over all concave utility profiles~$U$, concave utility functions $u_i'$, and voters $i$. 
	
	To prove that the incentive ratio of \nash is $2$ for concave utility functions, we first discuss three auxiliary claims. The proof of the following lemma is deferred to \Cref{app:Nash}.

	\begin{restatable}{lemma}{nashaux}\label{lem:nashaux}
		Fix an arbitrary concave utility profile $U$. 
		\begin{enumerate}[leftmargin=*, label=(\arabic*)]
			\item It holds that $u_i(\text{\nash}(U))>0$ for all voters $i\in N$. 
			\item For all distributions $q\in \Delta(C)$, it holds that 
			\[\sum_{i\in N} \frac{u_i(q)-u_i(\text{\nash}(U))}{u_i(\text{\nash}(U))}\leq 0.\]
			\item For all voters $i\in N$, it holds that \[e^{-1}\leq \frac{\prod_{j\in N\setminus \{i\}} u_j(\text{\nash}(U))}{\prod_{j\in N\setminus \{i\}} u_j(\text{\nash}(U_{-i}))}\leq 1.\] 
		\end{enumerate}
	\end{restatable}
	
	Less formally, Claim (1) shows that \nash guarantees each voter strictly positive utility. Claim (2) states the separation inequality witnessing that the distribution $p=\text{\nash}(U)$ maximizes the Nash welfare $\prod_{i\in N} u_i(p)$. While this inequality appears also in earlier works \citep[e.g.,][]{AAC+19a,ABM20a}, we derive it under more general conditions than these works as we do not even require that utility functions are differentiable. Lastly, Claim (3) demonstrates that the Nash welfare of the voters $N\setminus \{i\}$ under \nash cannot decrease too much when voter $i$ joins the election. A similar insight has been proven by \citet{CGG13a} for a private goods setting.
	
	Based on these insights, we now prove our main theorem.

	\begin{theorem}\label{thm:nash}
		The incentive ratio of \nash is $2$, even if voters are allowed to report concave utility functions. 
	\end{theorem}
	\begin{proof}
		We will show here only that the incentive ratio of \nash is at most $2$ since a matching lower bound follows from \Cref{thm:nashopt}. For this, we assume for contradiction that there is a concave utility profile $U$, a voter $i$, and a concave utility function $u_i'$ such that $\frac{u_i(\text{\nash}(U_{-i}, u_i'))}{u_i(\text{\nash}(U))}> 2$. To ease the notation, we let $p=\text{\nash}(U)$ and $q=\text{\nash}(U_{-i}, u_i')$. We~will show that $\sum_{j\in N} \frac{u_j(q)}{u_j(p)}>n$, which contradicts Claim~(2) in \Cref{lem:nashaux} as $\sum_{j\in N} \frac{u_j(q)-u_j(p)}{u_j(p)}=\sum_{j\in N} \frac{u_j(q)}{u_j(p)} -n$.
		
		To prove this inequality, we first derive a lower bound on $\sum_{j\in N\setminus \{i\}} \frac{u_j(q)}{u_j(p)}$. For this, we consider the term $\prod_{j\in N\setminus \{i\}} \frac{u_j(q)}{u_j(p)}$, which is well-defined by Claim (1) of \Cref{lem:nashaux}. Next, let $\hat p=\text{\nash}(U_{-i})$ denote the distribution chosen by \nash when voter $i$ is absent. Clearly, it holds that 
		$\prod_{j\in N\setminus \{i\}} \frac{u_j(q)}{u_j(p)}=\frac{\prod_{j\in N\setminus \{i\}} u_j(q)}{\prod_{j\in N\setminus \{i\}} u_j(\hat p)}\cdot \frac{\prod_{j\in N\setminus \{i\}} u_j(\hat p)}{\prod_{j\in N\setminus \{i\}} u_j(p)}$.
		Further, by Claim (3) of \Cref{lem:nashaux}, we have that $\frac{1}{e}\leq \frac{\prod_{j\in N\setminus \{i\}} u_j(q)}{\prod_{j\in N\setminus \{i\}} u_j(\hat p)}$ and $1\leq \frac{\prod_{j\in N\setminus \{i\}} u_j(\hat p)}{\prod_{j\in N\setminus \{i\}} u_j(p)}$. We hence conclude that 
		\[ \frac{1}{e}\leq
		\frac{\prod_{j\in N\setminus \{i\}} u_j(q)}{\prod_{j\in N\setminus \{i\}} u_j(\hat p)}\cdot \frac{\prod_{j\in N\setminus \{i\}} u_j(\hat p)}{\prod_{j\in N\setminus \{i\}}  u_j(p)}= \prod_{j\in N\setminus \{i\}} \frac{u_j(q)}{u_j(p)}.\]
		
		By applying the logarithm, it follows that $\log \frac{1}{e}\leq \sum_{j\in N\setminus \{i\}} \log \frac{u_j(q)}{u_j(p)}$.
		Next, Jensen's inequality implies that 
		\[\frac{1}{n-1}\sum_{j\in N\setminus \{i\}}\log  \frac{u_j(q)}{u_j(p)}\leq \log \left(\frac{1}{n-1}\sum_{j\in N\setminus \{i\}} \frac{u_j(q)}{u_j(p)}\right).\]
		
		By chaining our inequalities, we infer that 
		\[\log\frac{1}{e}\leq (n-1) \log \left(\frac{1}{n-1}\sum_{j\in N\setminus \{i\}} \frac{u_j(q)}{u_j(p)}\right).\]
		
		Solving for $\sum_{j\in N\setminus \{i\}} \frac{u_j(q)}{u_j(p)}$ then shows that
		\[(n-1) e^{-\frac{1}{n-1}}\leq \sum_{j\in N\setminus \{i\}} \frac{u_j(q)}{u_j(p)}.\]
		
		We now conclude that $\sum_{j\in N} \frac{u_j(q)}{u_j(p)}\geq (n-1) e^{-\frac{1}{n-1}} + 2$ as $\frac{u_i(q)}{u_i(p)}\geq 2$. We will next show that $(n-1)e^{-\frac{1}{n-1}} + 2>n$ for all $n\geq 2$, or equivalently that $(n-1)(e^{-\frac{1}{n-1}} -1)>-1$. (We also observe that \nash is strategyproof if $n=1$). For this, we note that $x(e^{-\frac{1}{x}}-1)>x(1-\frac{1}{x}-1)=-1$ for all $x>0$ since $e^{y}>1+y$ for all  $y\neq 0$. By setting $x=n-1$, this proves our inequality. Hence, the assumption that $\frac{u_i(q)}{u_i(p)}>2$ contradicts Claim (2) of \Cref{lem:nashaux}, thus proving our theorem. 
	\end{proof}
	
	\begin{remark}
		\Cref{thm:nash} generalizes to group manipulations: under \nash, no group of voters can misreport their utility functions such that each member of the group obtains more than twice her original utility. More formally, it holds for all concave utility profiles $U$ and groups of voters $S\subseteq N$ that there is no group manipulation $U_S'$ such that $\frac{u_i(\text{\nash}(U_{-S}, U_S'))}{u_i(\text{\nash}(U)}> 2$ for all $i\in S$. 
		This statement can be proven similar to \Cref{thm:nash} because Claim~(3) of \Cref{lem:nashaux} can be generalized to show that $e^{-|S|}\leq \frac{\prod_{i\in N\setminus S} u_i(\text{\nash}(U))}{\prod_{i\in N\setminus S} u_i(\text{\nash}(U_{-S}))}\leq 1$ for all groups of voters $S\subseteq N$. 
	\end{remark}

	\begin{remark}
		Based on Claim (2) of \Cref{lem:nashaux}, it can be shown that many of the desirable properties that \nash satisfies under approval preferences also hold under concave utility functions. In particular, by combining our lemma with the proof of Theorem 3 of \citet{ABM20a}, it follows that \nash is in the core for all concave utility functions \citep[see also][]{CLK+22a}. Thus, \nash is fair, efficient, and approximately strategyproof for all concave utility functions, whereas no rule satisfies these criteria when using full strategyproofness, even if only approval utilities are allowed \citep{BBPS21a}. 
	\end{remark}
	
	\subsection{Optimality of \nash}\label{subsec:Nashoptimality}
	
	A natural follow-up question to \Cref{thm:nash} is whether there are fair and efficient rules with a lower incentive ratio than \nash. As we will show next, such rules do not exist within three natural classes of distribution rules, even when only allowing for approval ballots. This suggests that the incentive ratio of \nash may be optimal subject to fairness and efficiency. Specifically, we will focus on the following classes of rules. 
	
	\paragraph{Rules satisfying AFS.} First, we consider rules that satisfy AFS (see \Cref{subsec:eff+fair}). By showing that each such rule has an incentive ratio of at least $2$, it follows that the incentive ratio of \nash is optimal subject to its fairness guarantees. 
	
	\paragraph{Strictly concave welfare maximizers (SCWMs).} As the second class, we study strictly concave welfare maximizers. To define these rules, we say a function $h:[0,1]\rightarrow \mathbb{R}\cup\{-\infty\}$ is a \emph{welfare function} if \emph{(i)} $h(x)\neq-\infty$ for all $x>0$ and \emph{(ii)} $h$ is strictly increasing. 
	Then, a distribution rule $f$ is an (additively separable) \emph{welfare maximizer} if there is a welfare function $h$ such that $f$ always returns a distribution maximizing the $h$-welfare, i.e., $f(A)\in\arg\max_{p\in \Delta(C)} \sum_{i\in N} h(u_i(x))$ for all profiles $A$. Further, $f$ is a \emph{strictly concave welfare maximizer (SCWM)} if $h$ is strictly concave (i.e., $h(\lambda x + (1-\lambda)y)>\lambda h(x)+(1-\lambda) h(y)$ for all distinct $x,y\in [0,1]$ and $\lambda\in (0,1)$). For example, this class contains all rules that optimize a power welfare $\sum_{i\in N} (u_i(p))^\alpha$ for~$\alpha \in(0,1)$. 
	
	We are interested in SCWMs since all rules in this class are efficient and many seem promising to obtain an incentive ratio of less than $2$. For example, one may conjecture that the SCWM defined by $h(x)=x^{1-\epsilon}$ for a very small $\epsilon>0$ has an incentive ratio less than $2$ as it is close to maximizing the utilitarian welfare. Our result for SCWMs disproves this conjecture and, more generally, shows that tweaks in the objective of \nash do not allow to improve its incentive ratio.
	
	\paragraph{Regular rules satisfying GFS and efficiency.} Thirdly, we will focus on rules that satisfy efficiency and GFS (see \Cref{subsec:eff+fair}).
	However, since both of these conditions are rather mild, we will additionally require three basic auxiliary conditions, which we summarize by regularity. Specifically, a distribution rule $f$ is \emph{regular} if it is
	\begin{itemize}[topsep=0pt, itemsep=0pt]
		\item \emph{anonymous}: $f(\pi(A))=f(A)$ for all profiles $A$ and permutations $\pi:N\rightarrow N$ (i.e., the identities of voters do not matter for the outcome),
		\item \emph{neutral}: $f(\tau(A),\tau(x))=f(A,x)$ for all profiles $A$ and permutations $\tau:C\rightarrow C$ (i.e., the names of candidates do not matter for the outcome), and
		\item \emph{independent of losers}: $f(A)=f(A')$ for all profiles $A,A'$ and candidates $x$ such that $f(A,x)=0$ and $A'=(A_i\setminus \{x\})_{i\in N}$ is obtained from $A$ by deleting $x$ from the approval ballot of every voter.
	\end{itemize}
	These three properties are very mild and satisfied by all rules in this paper (except for \map which fails neutrality). Our third claim in \Cref{thm:nashopt} thus shows that the incentive ratio of \nash is already optimal when only requiring efficiency, a mild fairness axiom, and some basic auxiliary conditions.
	\medskip
	
	We show next that every rule that belongs to one of these classes has an incentive ratio of $2$. Since \nash is a member of each class, all our bounds are tight. 
	
	\begin{restatable}{theorem}{nashopt}\label{thm:nashopt}
		The following claims are true:
		\begin{enumerate}[leftmargin=*, label=(\arabic*)]
			\item Every distribution rule that satisfies AFS has an incentive ratio of at least $2$. 
			\item Every SCWM has an incentive ratio of at least $2$. 
			\item Every regular distribution rule that satisfies GFS and efficiency has an incentive ratio of at least $2$.
		\end{enumerate}
	\end{restatable}
	\begin{proof}
	We prove each of the three claims separately. 
	\medskip

	\textbf{Proof of Claim (1)}:
	Let $f$ denote a rule satisfying AFS, and fix two integers ${\ell,k\in\mathbb{N}}$ with $k\geq \ell\geq 2$. We consider the following profile $A$ with $2k+\ell+1$ voters and $\ell+1$ candidates $C=\{x, y_1,\dots, y_\ell\}$: 
		\begin{itemize}[topsep=0pt, itemsep=0pt]
			\item For each $i\in [\ell]$, there is one voter who reports $\{x, y_i\}$. 
			\item There are $k+1$ voters who only approve $x$.
			\item There are $k$ voters who report $\{y_1,\dots, y_\ell\}$.
		\end{itemize}
		Next, the profile $A'$ is derived from $A$ by changing the approval ballot of the voter $i^*$ who reports $\{x,y_1\}$ in $A$ to $\{y_1\}$. For example, if $\ell=3$ and $k=8$, the profiles look as follows.\smallskip
		
		\noindent
		{\tabcolsep=3pt
			\begin{tabular}{lllllll}
				$A$: & $1$:$\,\{x, y_1\}$ & $1$:$\,\{x, y_2\}$ &  $1$:$\,\{x, y_3\}$ & $9$:$\,\{x\}$ & $8$:$\,\{y_1,y_2, y_3\}$\\
				$A'$: & $1$:$\,\{y_1\}$ & $1$:$\,\{x, y_2\}$ &  $1$:$\,\{x, y_3\}$ & $9$:$\,\{x\}$ & $8$:$\,\{y_1,y_2, y_3\}$
			\end{tabular}
		}
		
		We first analyze the distribution $p=f(A)$. To this end, we observe that AFS implies for the $k$ voters who report $\{y_1,\dots, y_\ell\}$ that $\frac{k}{k}\sum_{i\in [\ell]} p(y_i)\geq \frac{k}{2k+\ell+1}$. Without loss of generality, we assume that $y_1$ is the candidate with the minimal probability in $\{y_1,\dots, y_\ell\}$, which means that $p(y_1)\leq \frac{1}{\ell} \sum_{i\in [\ell]} p(y_i)$ and $\sum_{i\in [\ell]\setminus \{1\}} p(y_i)\geq \frac{\ell-1}{\ell} \sum_{i\in[\ell]} p(y_i)\geq \frac{\ell-1}{\ell}\cdot \frac{k}{2k+\ell+1}$. In turn, this implies for the utility of voter $i^*$ (who reports $\{x, y_1\}$ in~$A$) that $u_{i^*}(p)\leq 1-\frac{\ell-1}{\ell}\cdot \frac{k}{2k+\ell+1}$. 
		
		Next, we turn to $q=f(A')$. By applying AFS to the sets $N_x=\{i\in N\colon x\in A_i'\}$ and $N_{y_1}=\{i\in N\colon y_1\in A_i'\}$, respectively, we infer the following two inequalities for $q$. 
		\begin{align*}
			\sum_{i\in N_x} u_i(q)&=(k+\ell)q(x)+\sum_{i\in [\ell]\setminus \{1\}} q(y_i)\geq \frac{(k+\ell)^2}{2k+\ell+1}\\
			\sum_{i\in N_{y_1}} u_i(q)&=(k+1)q(y_1)+k\!\!\!\sum_{i\in [\ell]\setminus \{1\}}\!\!\! q(y_i)\geq \frac{(k+1)^2}{2k+\ell+1}
		\end{align*}
		By summing up these two inequalities, we derive that 
		\[(k+\ell)q(x)+(k+1)\sum_{i\in [\ell]} q(y_i)\geq\frac{(k+\ell)^2+(k+1)^2}{2k+\ell+1}.\]
		Since $q(x)+\sum_{i\in[\ell]} q(y_i)=1$, this means that 
		\begin{align*}
			(\ell-1)q(x)&\geq \frac{(k+\ell)^2+(k+1)^2-(k+1)(2k+\ell+1)}{2k+\ell+1}\\
			&=\frac{k(\ell-1)+\ell(\ell-1)}{2k+\ell+1}.
		\end{align*}
		
		This shows that $q(x)\geq \frac{k+\ell}{2k+\ell+1}$. In turn, we derive that $\sum_{i\in [\ell]} q(y_i)\leq \frac{k+1}{2k+\ell+1}$ and by the AFS inequality for $y_1$ that $q(y_1)\geq \frac{k+1}{2k+\ell+1}$. Since voter $i^*$ reports $\{x, y_i\}$ in $A$, this proves that $u_{i^*}(q)=1$. 
		Hence, the incentive ratio of $f$ is at least $\frac{u_{i^*}(q)}{u_{i^*}(p)}\geq\frac{1}{1-\frac{\ell-1}{\ell}\cdot \frac{k}{2k+\ell+1}}$. 
		When setting $k=\ell^2$ and letting $\ell$ go to infinity, this fraction converges to $2$, so every rule satisfying AFS has an incentive ratio of at least~$2$.\medskip

	\textbf{Proof of Claim (2)}: Fix a strictly concave welfare function $h$ and let $f$ denote the induced SCWM. Further, we fix an integer $\ell\in \mathbb{N}$ with $\ell\geq 2$, and consider the following profile $A$ with $2\ell$ voters and $\ell+2$ candidates $C=\{x_1,x_2,y_1,\dots, y_\ell\}$. 
		\begin{itemize}[topsep=0pt, itemsep=0pt]
			\item For each $i\in [\ell]$, there is one voter who approves $x_1$ and all candidates $\{y_1,\dots, y_\ell\}$ except $y_i$. Put differently, these voters approve $C\setminus\{x_2,y_i\}$ for $i\in [\ell]$. 
			\item There is one voter who approves all candidates but $x_1$. We denote this voter by $i^*$. 
			\item There are $\ell-1$ voters who only approve $x_2$. 
		\end{itemize}
		Further, the profile $A'$ is derived from $A$ by changing the ballot of voter $i^*$ to $C\setminus \{x_1,x_2\}$, i.e., this voter additionally disapproves $x_2$. We subsequently illustrate the profiles $A$ and~$A'$ for $\ell=3$.
		
		\noindent
		\begin{tabular}{llllllllll}
			$A$: & $1$: $\{x_1,y_2,y_3\}$ & $1$: $\{x_1,y_1,y_3\}$ & $1$: $\{x_1,y_1,y_2\}$ \\
			&$1$: $\{x_2,y_1,y_2,y_3\}$ & $2$: $\{x_2\}$\smallskip\\
			$A'$: & $1$: $\{x_1,y_2,y_3\}$ & $1$: $\{x_1,y_1,y_3\}$ & $1$: $\{x_1,y_1,y_2\}$ \\
			& $1$: $\{y_1,y_2,y_3\}$ & $2$: $\{x_2\}$
		\end{tabular}

		We claim that $f$ chooses for $A$ the distribution $p$ with $p(x_1)=p(x_2)=0.5$. To prove this, we first observe that every candidate is approved by exactly half of the voters in $A$. Hence, for every distribution $p'$, it holds that $\sum_{i\in N} u_i(p')=\sum_{x\in C} \sum_{i\in N\colon x\in A_i} p'(x)=\ell \sum_{x\in C} p'(x)=\ell$. Now, assume for contradiction that $f$ chooses a distribution $p'$ for $A$ such that $u_i(p')\neq 0.5$ for some voter $i\in N$. By the strict concavity of $h$, we have that 
		\begin{align*}
			\sum_{i\in N} h(0.5)
			=\sum_{i\in N} h \left(\frac{1}{2\ell} \sum_{j\in N} u_j(p') \right)
			>\sum_{i\in N} \frac{1}{2\ell} \sum_{j\in N} h(u_j(p'))
			=\sum_{i\in N} h(u_j(p')).
		\end{align*}
		This proves that $p$ generates a larger $h$-welfare than $p'$, so $p'$ cannot be chosen by $h$. Finally, to ensure that every voter has a utility of $0.5$, we must have $p(x_2)=0.5$ because there are $\ell-1$ voters who only approve $x_2$. Further, since all voters need to have utility of $0.5$, this implies that $p(y_i)=0$ for all $i\in [\ell]$ as voter $i^*$ approves all candidates but $x_1$. In turn, this proves that $p(x_1)=0.5$. 
		
		Next, for $A'$, we claim that $f$ has to choose a distribution $q$ with $q(x_1)=q(y_1)=\dots=q(y_\ell)$. To see this, we set $y_0=x_1$ and define $a+_{\ell+1} b = (a+b)\mod (\ell+1)$ for all $a,b\in \mathbb{N}$. Now, fix a distribution $q'$ that violates our equation, and denote by $q^k$ for $k\in \{0,\dots, \ell\}$ the distribution with $q^k(x_2)=q'(x_2)$ and $q^k(y_j)=q'(y_{j+_{\ell+1} k})$ for all $j\in \{0,\dots, \ell\}$. We note that, for all candidates $y_i,y_j$, it holds that $\sum_{k=0}^\ell q^k(y_i)=\sum_{k=0}^\ell q'(y_k) = \sum_{k=0}^\ell q^k(y_j)$. Lastly, let $q^*$ be the distribution defined by $q^*=\frac{1}{\ell+1}\sum_{k=0}^\ell q^k$ and observe that $q^*(y_i)=q^*(y_j)$ for all $i,j\in \{0,\dots, \ell\}$. This also means that $q^*\neq q'$ as $q'$ fails this condition.
		
		We will show that $\sum_{i\in N} h(u_i(q^*))>\sum_{i\in N} h(u_i(q'))$, which proves that $f$ cannot choose $q'$ for $A'$. First, for the voters who approve only $x_2$, we clearly have that $u_i(q')=u_i(q^*)$ since $q'(x_2)=q^*(x_2)$. On the other hand, for the remaining $\ell+1$ voters, we observe that, for every candidate $y_i\in \{y_0=x_1,y_1\dots, y_\ell\}$, there is exactly one voter who reports $C\setminus \{y_i,x_2\}$. This means that, for every voter $i$ and integer $k$, there is precisely one voter $j$ such that $u_i(q^k)=u_j(q^0)$.  Hence, when letting $S$ denote the set of these $\ell+1$ voters, we have for each $i\in S$ that $u_i(q^*)=\frac{1}{\ell+1}\sum_{j\in S} u_j(q^0)$. Using the strict concavity of $h$, we compute that 
		
		\begin{align*}
			\sum_{i\in S} h(u_i(q^*))
			= \sum_{i\in S} h \left(  \frac{1}{\ell+1}\sum_{j\in S} u_j(q^0) \right)
			>\sum_{i\in S} \frac{1}{\ell+1}\sum_{j\in S} h(u_j(q^0))
			=\sum_{i\in S} h(u_i(q')). 
		\end{align*}
		In summary, we conclude that $\sum_{i\in N} h(u_i(q^*))> \sum_{i\in N} h(u_i(q'))$, as desired. Consequently, the distribution $q$ chosen for $A'$ satisfies that $q(x_1)=q(y_1)=\dots=q(y_\ell)$. 
		
		Finally, by our analysis so far, we conclude that voter $i^*$ (who reports $C\setminus \{x_1\}$ in $A$) has utility $u_i(p)=\frac{1}{2}$ in $A$ and utility $u_i(q)=q(x_2)+\frac{\ell}{\ell+1} (1-q(x_2))=\frac{\ell}{\ell+1}+\frac{1}{\ell+1}q(x_2)$ in $A'$. For the second expression, we use that the probability not assigned to $x_2$ by $q$ is spread uniformly over the remaining $\ell+1$ candidates, and that voter $i^*$ approves $\ell$ of them. As a consequence, the incentive ratio of $f$ is lower bounded by \[\frac{u_i(q)}{u_i(p)}=\frac{\frac{\ell}{\ell+1}+\frac{1}{\ell+1}q(x_2)}{\frac{1}{2}}\geq \frac{2\ell}{\ell+1}.\]
		Finally, by letting $\ell$ grow to infinity, this term converges to $2$, which proves that the incentive ratio of $f$ is at least $2$.\medskip
		
		\textbf{Proof of Claim (3):}
		Let $f$ denote a rule that satisfies efficiency, group fair share, anonymity, neutrality, and independence of losers. Further, we fix some value $k\in \mathbb{N}$ with $k\geq 3$, define the set of voters by $N=[k+1]^2$, and the set of candidates by $C=\{x_i\colon i\in [k+1]\}\cup \{y_{i,j}\colon i,j\in [k]\}\cup \{z\}$. For this proof, we again denote voters by tuples $(i,j)$, where $i$ and $j$ are best interpreted as row and column indexes, respectively. Further, to simply notation, we define $X=\{x_i\colon i\in [k+1]\}$ and $Y=\{y_{i,j}\colon i,j\in [k]\}$. Lastly, for our proof, we will focus on three profiles. We next describe these profiles formally and note that they are illustrated in \Cref{fig:Claim3profs} for $k=3$. First, the main profile $A^1$ is constructed as follows:
		\begin{itemize}[topsep=0pt, itemsep=0pt]
			\item Every voter $(i,j)\in [k]^2$ approves candidate $x_i$ and $y_{i,j}$.
			\item For all $i\in [k]$ with $i>1$, each voter $(i, k+1)$ approves $\{z\}\cup Y\setminus \{y_{j,i}\colon j\in [k]]\}$. Further, voters $(1,k+1)$ and $(k+1, k+1)$ approve $X$.
			\item For all $j\in [k]$ with $j>1$, each voter $(k+1,j)$ only approves $x_{k+1}$. Further, voter $(1,k+1)$  reports $\{x_{k+1},z\}$.\smallskip
		\end{itemize}
		
		\begin{figure*}
			\centering\tabcolsep=4pt
			\scalebox{0.9}{
				\begin{tabular}{l|cccc}
					$A^1$: & $j=1$ & $j=2$ & $j=3$ & $j=4$\\\hline
					$i=1$ & $\{x_1, y_{1,1}\}$ &  $\{x_1, y_{1,2}\}$ &  $\{x_1, y_{1,3}\}$ & $\{x_1,x_2,x_3,x_4\}$\\
					$i=2$ & $\{x_2, y_{2,1}\}$ &  $\{x_2, y_{2,2}\}$ &  $\{x_2, y_{2,3}\}$ & $\{z, y_{1,1}, y_{2,1}, y_{3,1}, y_{1,3}, y_{2,3}, y_{3,3}\}$\\
					$i=3$ & $\{x_3, y_{3,1}\}$ &  $\{x_3, y_{3,2}\}$ &  $\{x_3, y_{3,3}\}$ & $\{z, y_{1,1}, y_{2,1}, y_{3,1}, y_{1,2}, y_{2,2}, y_{3,2}\}$\\
					$i=4$ & $\{z,x_4\} $ & $\{x_4\}$ & $\{x_4\}$ &  $\{x_1,x_2,x_3,x_4\}$
				\end{tabular}
				\hspace{0.5cm}
				\begin{tabular}{l|cccc}
					$A^2$: & $j=1$ & $j=2$ & $j=3$ & $j=4$\\\hline
					$i=1$ & $\{x_1, y_{1,1}\}$ &  $\{x_1\}$ &  $\{x_1\}$ & $\{x_1,x_2,x_3,x_4\}$\\
					$i=2$ & $\{x_2, y_{2,1}\}$ &  $\{x_2\}$ &  $\{x_2\}$ & $\{z, y_{1,1}, y_{2,1}, y_{3,1}\}$\\
					$i=3$ & $\{x_3, y_{3,1}\}$ &  $\{x_3\}$ &  $\{x_3\}$ & $\{z, y_{1,1}, y_{2,1}, y_{3,1}\}$\\
					$i=4$ & $\{z,x_4\} $ & $\{x_4\}$ & $\{x_4\}$ &  $\{x_1,x_2,x_3,x_4\}$
			\end{tabular}}\bigskip
			
			\scalebox{0.9}{
				\begin{tabular}{l|cccc}
					& $j=1$ & $j=2$ & $j=3$ & $j=4$\\\hline
					$i=1$ & $\{x_1, y_{1,1}\}$ &  $\{x_1, y_{1,2}\}$ &  $\{x_1, y_{1,3}\}$ & $\{x_1,x_2,x_3,x_4\}$\\
					$i=2$ & $\{x_2, y_{2,1}\}$ &  $\{x_2, y_{2,2}\}$ &  $\{x_2, y_{2,3}\}$ & $\{z, y_{1,1}, y_{2,1}, y_{3,1}, y_{1,3}, y_{2,3}, y_{3,3}\}$\\
					$i=3$ & $\{x_3, y_{3,1}\}$ &  $\{x_3, y_{3,2}\}$ &  $\{x_3, y_{3,3}\}$ & $\{z, y_{1,1}, y_{2,1}, y_{3,1}, y_{1,2}, y_{2,2}, y_{3,2}\}$\\
					$i=4$ & $\{z, y_{1,2}, y_{2,2}, y_{3,2}, y_{1,3}, y_{2,3}, y_{3,3}\}$ & $\{x_4\}$ & $\{x_4\}$ &  $\{x_1,x_2,x_3,x_4\}$
			\end{tabular}}
			\caption{Profiles for the proof of Claim (3) of \Cref{thm:nashopt}}
			\label{fig:Claim3profs}
		\end{figure*}
		
		We start by analyzing the distribution $p^1$ chosen by $f$ for $A^1$. To this end, we observe that for all $i\in [k]$, $j,j'\in [k]\setminus\{1\}$, the candidates $y_{i,j}$ and $y_{i,j'}$ are symmetric to each other in $A^1$. To make this more formal, we define the permutation $\pi^{j,j'}:C\rightarrow C$ for all $j,j'\in [k]\setminus\{1\}$ by \emph{(i)} $\pi^{j,j'}(y_{i,j})=y_{i,j'}$ and $\pi^{j,j'}(y_{i,j'})=y_{i,j}$ for all $i\in [k]$ and \emph{(ii)} $\pi^{j,j'}(c)=c$ for all other candidates $c$. Applying this permutation to $A^1$ has two effects: firstly, we exchange the $j$-th and $j'$-th columns, and, secondly, the voters $(j,k+1)$ and $(j',k+1)$ swap their ballots. Hence, up to renaming voters, the resulting profile is equal to $A^1$. By anonymity and neutrality, it thus follows that $p^1(y_{i,j})=p^1(y_{i,j'})$ for all $i\in [k]$ and $j,j'\in [k]\setminus \{1\}$. 
		
		Next, we will show that $p^1(y_{i,j})=0$ for all $i\in [k]$ and $j\in [k]\setminus\{1\}$. Assume for contradiction that this is not true, i.e., there are indices $i\in [k]$, $j\in [k]\setminus\{1\}$ such that $p^1(y_{i,j})=\epsilon>0$. By our previous observation, this means that $p^1(y_{i,j'})=\epsilon$ for all $j'\in [k]\setminus \{1\}$. We will show that $p^1$ is inefficient. To this end, we define the distribution $q$ by $q(x_i)=p^1(x_i)+\epsilon$, $q(y_{i,1})=p^1(y_{i,1})+(k-2)\epsilon$, $q(y_{i,j'})=0$ for all $j'\in [k]\setminus \{1\}$, and $q(c)=p^1(c)$ for all other candidates.
		We claim that $u_{i',j'}(q)\geq u_{i',j'}(p^1)$ for all voters and $u_{i',j'}(q)> u_{i',j'}(p^1)$ for at least one voter. To see the second claim, we note that 
		\begin{align*}
			u_{1,k+1}(q)&=\!\sum_{\ell \in [k+1]} q(x_\ell)
			=\!\epsilon+\sum_{\ell \in [k+1]} p^1(x_\ell)
			>u_{1,k+1}(p^1).
		\end{align*}
		
		On the other hand, to see that the utility of no voter decreases, we first observe that $u_{i',j'}(q)\geq u_{i',j'}(p^1)$ for all voters who do not approve any candidate $y_{i,j}$ with $j\in [k]\setminus \{1\}$. This leaves us with the voters $(i,\ell)$ for $\ell\in [k]\setminus \{1\}$ (i.e., those in the $i$-th row) and the voters $(\ell, k+1)$ for $\ell\in [k]\setminus \{1\}$ (i.e., those in the $k+1$-th column). For the voters $(i,\ell)$ with $\ell\in [k]\setminus \{1\}$, it holds that 
		\[u_{i,\ell}(q)=q(x_i)+q(y_{i,\ell})=p^1(x_i)+\epsilon+0=u_{i,\ell}(p^1). \]
		On the other hand, for the voters $(\ell, k+1)$ with $\ell\in [k]\setminus\{1\}$, we note that all of these voters approve $y_{i,1}$ and exactly $k-2$ candidates of $\{y_{i,2},\dots, y_{i,k}\}$. Thus, the utility of these voters satisfies that 
		\begin{align*}
			u_{\ell,k+1}(q)&=\sum_{c\in A_{\ell,k+1}^1} q(c)\\
			&=(k-2)\epsilon- \sum_{j'\in [k]\setminus \{\ell\}}  p^1(y_{i,j'})+\sum_{c\in A_{\ell,k+1}^1} p^1(c)\\
			&=u_{\ell,k+1}(p^1).
		\end{align*}
		This proves that $q$ indeed Pareto-dominates $p^1$, so our assumption that $p^1(y_{i,j})>0$ must have been wrong.  Put differently, this means that $p^1(y_i,j)=0$ for all $i\in [k]$, $j\in [k]\setminus\{1\}$. 
		
		Now, let $A^2$ denote the profile derived from $A^1$ by deleting all candidates $y_{i,j}$ with $i\in [k]$, $j\in [k]\setminus \{1\}$ from our profile and let $p^2$ denote the distribution chosen by $f$ for this profile. An example of $A^2$ for $k=3$ is depicted in \Cref{fig:Claim3profs}. It is easy to see that, in $A^2$, all candidates in $X$ and the candidates in $Y'=\{y_{1,1},\dots, y_{k,1},z\}$ are symmetric to each other, respectively. 
		In particular, this means that $p^2(x_1)=\dots=p^2(x_{k+1})$ and $p^2(y_{1,1})=\dots=p^2(y_{1,k})=p^2(z)$. For voter $(k+1,1)$, this means that $u_{k+1,1}(p^2)=p^2(z)+p^2(x_{k+1})=\frac{1}{k+1}\sum_{c\in X} p^2(c)+\frac{1}{k+1}\sum_{c\in Y'} p^2(c)=\frac{1}{k+1}$. Further, since we derived $A^2$ from $A^1$ by only deleting candidates with probability $0$ in $p^1$, independence of losers requires that $p^1=p^2$ and thus $u_{k+1,1}(p^1)=\frac{1}{k+1}$, too.
		
		Lastly, consider the profile $A^3$ derived from $A^1$ by assigning voter $(k+1, 1)$ the ballot $z\cup Y\setminus \{y_{1,1},\dots, y_{k,1}\}$ and let $p^3=f(A^3)$; an example of this profile is again shown in \Cref{fig:Claim3profs}.
		Using an analogous argument as for $A^1$, it can be shown that $p^3(y_{i,j})=p^3(y_{i,j'})$ for all $i,j,j'\in [k]$ based on anonymity and neutrality. Note that, in contrast to $A^1$, this inequality now also includes the candidates $y_{i,1}$. We claim that this means that $p^3(y_{i,j})=0$ for all $i,j\in [k]$. Assume for contradiction that there are some indices $i,j\in [k]$ such that $p^3(y_{i,j})>0$ and let $p^3(y_{i,j})=\epsilon$. By our previous observation, this implies that $p^3(y_{i,j'})=\epsilon>0$ for all $j'\in [k]$. Now, consider the distribution $q$ defined by $q(x_i)=p^3(x_i)+\epsilon$, $q(z)=p^3(z)+(k-1)\epsilon$, $q(y_{i,j'})=0$ for all $j'\in [k]$, and $q(c)=p^3(c)$ for all remaining candidates. First, we note that voter $(1,k+1)$ strictly prefers $q$ to $p^3$ since $\sum_{x\in X} q(x)=\epsilon +\sum_{x\in X} p^3(x)$. Next, for every voter who does not approve any candidate in $\{y_{i,1},\dots, y_{i,k}\}$, it holds that $u_{i',j'}(q)\geq u_{i',j'}(p^3)$ because we only reduce the shares of the candidates in this set. This leaves us with the voters in the $i$-th row (i.e,. $(i,\ell)$ with $\ell\in [k]$), the voters in the $(k+1)$-th column (i.e., $(\ell, k+1)$ with $\ell\in [k]\setminus\{1\}$), and voter $(k+1,1)$. Now, for all voters $(i,\ell)$ with $\ell\in[k]$, it holds that 
		\begin{align*}
			u_{i,\ell}(q)=q(x_i)+q(y_{i,\ell}) = p^3(x_i)+\epsilon+0=u_{i,\ell}(p^3).
		\end{align*}
		Further, for the the voters $(\ell, k+1)$ with $\ell\in [k]\setminus\{1\}$ and voter $(k+1,1)$, it holds that they approve precisely $k-1$ candidates in $Y_i=\{y_{i,1},\dots, y_{i,k}\}$. Hence, we get for all these voters that 
		\begin{align*}
			u_{\ell,k+1}(q)=\sum_{c\in A_{\ell,k}^1} q(c)
			\ = \ (k-1)\epsilon \ - \ \sum_{y_{i,j'}\in Y_i\cap A_{\ell,k+1}^1} \epsilon \ + \ \sum_{c\in A_{\ell,k+1}^1} p^3(c)
			\ = \ u_{\ell, k+1}(p^3).
		\end{align*}
		This proves that $q$ indeed dominates $p^3$, which contradicts the efficiency of $f$. Hence, it holds indeed that $p^3(y_{i,j})=0$ for all $i,j\in [k]$.
		
		Lastly, GFS for the set of voters $S=\{(k+1,j)\colon j\in [k]\}\cup \{(i,k+1)\colon i\in [k]\setminus \{1\}\}$ requires that the total share allocated to approved candidates of these voters is at least $\frac{|S|}{n}$. Since $p^3(y)=0$ for all $y\in Y$, this means that $p^3(x_{k+1})+p^3(z)\geq \frac{|S|}{n}=\frac{2k-1}{(k+1)^2}$. Hence, it holds that $u_{k+1,1}(p^3)\geq \frac{2k-1}{(k+1)^2}$ for the true utility function of voter $(k+1,1)$. 
		
		This means for the incentive ratio of $f$ is lower bounded by 
		\begin{align*}
			\frac{u_{k+1,1}(p^3)}{u_{k+1,1}(p^1)}=\frac{(2k-1)/(k+1)^2}{1/(k+1)}=\frac{2k-1}{k+1}. 
		\end{align*}
		
		For $k\rightarrow\infty$, this term approaches $2$, thus showing that the incentive ratio of $f$ is at least $2$.
	\end{proof}
	
	\begin{remark}
		It is possible to design distribution rules with an incentive ratio of $\gamma$ for every $\gamma\in(1,2]$ by mixing \nash with the uniform distribution. Specifically, consider the rule $f_\lambda(A)=\lambda \text{\nash}(A)+(1-\lambda)p_\mathit{uni}$ for $\lambda\in (0,1)$, which returns the convex combination of \nash and the uniform distribution $p_\mathit{uni}$ over all candidates. For example, when setting $\lambda=\frac{1}{m+1}$, one can show that $f_\lambda$ has an incentive ratio between $\frac{6}{5}$ and $\frac{4}{3}$, where the lower bounds follows by using the example of Claim (1) of \Cref{thm:nashopt}. However, to obtain a notable reduction in the incentive ratio, one needs to use a very small $\lambda$, so we lose the desirable properties of \nash. 
	\end{remark}
	
	\section{Experiments} 
	
	As our last contribution, we experimentally evaluate the manipulability of the four distribution rules in \Cref{subsec:rules}. In particular, we are interested in how often and how severely our distribution rules are manipulable in more realistic profiles. 
	
	To answer this question, we sample $1000$ approval profiles for each number of voters $n\in \{10,20,\dots,100\}$ and $m=10$ candidates. For each rule $f\in \{\text{\nash}, \text{\egal}, \text{\map}, \text{\fut}\}$ and each sampled profile $A$, we then compute the maximal multiplicative utility gain of a voter, defined by $IR(f,A)=\max_{i, A_i'} \frac{u_i(f(A_{-i}, A_i'))}{u_i(f(A))}$. We note that a rule $f$ is manipulable in a profile $A$ if and only if $IR(f,A)>1$. We repeat these simulations with the following three distributions. 
	
	\paragraph{Impartial culture model.} Each voter approves each candidate with probability $p=0.3$ independently at random. 
	
	\paragraph{Euclidean model.} Voters and candidates are placed uniformly at random in the 3d unit ball and voters approve all candidates with $\ell_2$-distance at most $0.4$ from their~position. 
	
	\paragraph{Truncated Mallow's model} In this model, we sample for each voter $i$ a strict preference relation $\succ_i$ from Mallow's model with a dispersion parameter of $\phi=0.75$ and a threshold $t_i\in \{1,\dots, m-1\}$ from a uniform distribution. Each voter approves the first $t_i$ candidates in her preference relation $\succ_i$.\smallskip

	Our findings of these experiments are summarized in \Cref{fig:experimenst1,fig:experimenst2}. Specifically, in \Cref{fig:experimenst1} we show the average incentive ratio and the percentage of manipulable profiles of our rules when $n=100$ and $m=10$. Further, we display the average incentive ratio of our rules as a function of~$n$ for the euclidean model in 
	\Cref{fig:experimenst2}. More details, including additional discussions and statistics, can be found in \Cref{app:experiments}.

	Maybe the most striking insight of our experiments is that, for all distributions, all our rules are frequently manipulable, with \egal and \nash being almost always manipulable. A possible explanation for this is that each of our rules satisfies some fairness axioms, thereby guaranteeing that each voter can influence the outcome and thus also manipulate. 
	Secondly, we note that, while \nash is frequently manipulable, its incentive ratio is very low with an average of consistently less than~$1.2$. This shows that voters indeed have little incentive to manipulate under \nash. Surprisingly, the incentive ratio of \fut and its percentage of manipulable profiles are even less in our experiments. Hence, while \fut theoretically allows for severe manipulations, the corresponding instances may not occur in practice. However, we note that the maximum incentive ratio of \fut in our simulations (typically between $2$ and $3$) is usually higher than that of \nash (typically around~$1.5$).
	Lastly, both \map and \egal are severely manipulable in our experiment. For \egal, the reason is that, in almost all profiles, some voter can notably increase her utility via free-riding (i.e., only voting for approved candidates that obtained no share before the manipulation). By contrast, the average incentive ratio of \map is driven by instances where it is extremely manipulable. For example, we even encountered a profile $A$ with $IR(\text{\map}, A)=37$ in our experiments. These large incentive ratios occur consistently for \map as even the $90\%$ percentile of the incentive ratio is often very large, which suggests that this rule can be severely manipulated on realistic profiles.
	
	\begin{figure}
		\setlength{\tabcolsep}{0.5em}	
		\centering
		\pgfplotstableread[col sep=semicolon]{./simulations/EUC.csv}\euctable
		\pgfplotstableread[col sep=semicolon]{./simulations/IC.csv}\unitable
		\pgfplotstableread[col sep=semicolon]{./simulations/MALL.csv}\malltable
		\pgfkeys{/pgf/number format/.cd, assume math mode=false}
		
		\begin{tabular}{c cc cc cc}
			\toprule
			& \multicolumn{2}{c}{Impartial} & \multicolumn{2}{c}{Euclidean}  & \multicolumn{2}{c}{Mallows}\\
			Rule & Freq & IR & Freq & IR& Freq & IR\\
			\midrule 
			\nash 
			& \pgfplotstablegetelem{9}{nash_freq}\of{\unitable}
			\pgfmathparse{100*\pgfplotsretval}
			\pgfmathprintnumber[fixed,precision=0]{\pgfmathresult}\%
			& \pgfplotstablegetelem{9}{nash_avg}\of{\unitable}
			\pgfmathprintnumber[fixed,precision=2]{\pgfplotsretval}
			& \pgfplotstablegetelem{9}{nash_freq}\of{\euctable}
			\pgfmathparse{100*\pgfplotsretval}
			\pgfmathprintnumber[fixed,precision=0]{\pgfmathresult}\%
			& \pgfplotstablegetelem{9}{nash_avg}\of{\euctable}
			\pgfmathprintnumber[fixed,precision=2]{\pgfplotsretval}
			& \pgfplotstablegetelem{9}{nash_freq}\of{\malltable}
			\pgfmathparse{100*\pgfplotsretval}
			\pgfmathprintnumber[fixed,precision=0]{\pgfmathresult}\%
			& \pgfplotstablegetelem{9}{nash_avg}\of{\malltable}
			\pgfmathprintnumber[fixed,precision=2]{\pgfplotsretval}
			\\
			\egal 
			&\pgfplotstablegetelem{9}{egal_freq}\of{\unitable}
			\pgfmathparse{100*\pgfplotsretval}
			\pgfmathprintnumber[fixed,precision=0]{\pgfmathresult}\%
			&\pgfplotstablegetelem{9}{egal_avg}\of{\unitable}
			\pgfmathprintnumber[fixed,precision=2]{\pgfplotsretval} 
			&  \pgfplotstablegetelem{9}{egal_freq}\of{\euctable}
			\pgfmathparse{100*\pgfplotsretval}
			\pgfmathprintnumber[fixed,precision=0]{\pgfmathresult}\%
			&\pgfplotstablegetelem{9}{egal_avg}\of{\euctable}
			\pgfmathprintnumber[fixed,precision=2]{\pgfplotsretval} 
			& \pgfplotstablegetelem{9}{egal_freq}\of{\malltable}
			\pgfmathparse{100*\pgfplotsretval}
			\pgfmathprintnumber[fixed,precision=0]{\pgfmathresult}\%
			&\pgfplotstablegetelem{9}{egal_avg}\of{\malltable}
			\pgfmathprintnumber[fixed,precision=2]{\pgfplotsretval} 
			\\
			\fut 
			& \pgfplotstablegetelem{9}{fut_freq}\of{\unitable}
			\pgfmathparse{100*\pgfplotsretval}
			\pgfmathprintnumber[fixed,precision=0]{\pgfmathresult}\%
			&\pgfplotstablegetelem{9}{fut_avg}\of{\unitable}
			\pgfmathprintnumber[fixed,precision=2]{\pgfplotsretval} 
			&  \pgfplotstablegetelem{9}{fut_freq}\of{\euctable}
			\pgfmathparse{100*\pgfplotsretval}
			\pgfmathprintnumber[fixed,precision=0]{\pgfmathresult}\%
			&\pgfplotstablegetelem{9}{fut_avg}\of{\euctable}
			\pgfmathprintnumber[fixed,precision=2]{\pgfplotsretval} 
			&  \pgfplotstablegetelem{9}{fut_freq}\of{\malltable}
			\pgfmathparse{100*\pgfplotsretval}
			\pgfmathprintnumber[fixed,precision=0]{\pgfmathresult}\%
			&\pgfplotstablegetelem{9}{fut_avg}\of{\malltable}
			\pgfmathprintnumber[fixed,precision=2]{\pgfplotsretval} 
			\\
			\map 
			&  \pgfplotstablegetelem{9}{mp_freq}\of{\unitable}
			\pgfmathparse{100*\pgfplotsretval}
			\pgfmathprintnumber[fixed,precision=0]{\pgfmathresult}\%
			& \pgfplotstablegetelem{9}{mp_avg}\of{\unitable}
			\pgfmathprintnumber[fixed,precision=2]{\pgfplotsretval}
			& \pgfplotstablegetelem{9}{mp_freq}\of{\euctable}
			\pgfmathparse{100*\pgfplotsretval}
			\pgfmathprintnumber[fixed,precision=0]{\pgfmathresult}\%
			& \pgfplotstablegetelem{9}{mp_avg}\of{\euctable}
			\pgfmathprintnumber[fixed,precision=2]{\pgfplotsretval}
			& \pgfplotstablegetelem{9}{mp_freq}\of{\malltable}
			\pgfmathparse{100*\pgfplotsretval}
			\pgfmathprintnumber[fixed,precision=0]{\pgfmathresult}\%
			& \pgfplotstablegetelem{9}{mp_avg}\of{\malltable}
			\pgfmathprintnumber[fixed,precision=2]{\pgfplotsretval}
			\\
			\bottomrule
		\end{tabular}

		\caption{Summary of our experiments. For each rule, we show the percentage of manipulable profiles (Freq) and the average incentive ratio (IR, rounded to two decimal digits) for our three sampling models when $n=100$ and $m=10$. }
		\label{fig:experimenst1}
	\end{figure}
	
	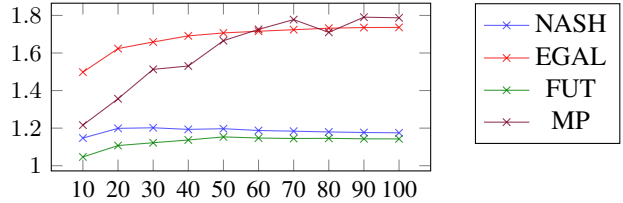
\begin{figure}
		\centering
		\pgfplotstableread[col sep=semicolon]{./simulations/EUC.csv}\datatable
		\edef\nashcol{nash_avg}%
		\edef\egalcol{egal_avg}%
		\edef\futcol{fut_avg}%
		\edef\mpcol{mp_avg}%
		\begin{tikzpicture}
			\begin{axis}[
				width=6.6cm,
				height=3.8cm,
				x tick label style={font=\small},
				xtick=data,
				xticklabels from table={\datatable}{n},
				legend style={at={(1.5,1)},anchor=north east},
				y tick label style={font=\small,
				}
				]    
				\addplot [name path = nash, mark=x, blue!80, opacity=0.8] table [x expr=\coordindex, y=\nashcol]{\datatable};
				\addlegendentry{NASH};
				
				\addplot [name path = egal, mark=x, red!95!black, opacity=0.8] table [x expr=\coordindex, y=\egalcol]{\datatable};
				\addlegendentry{EGAL};

				\addplot [name path = fut, mark=x, green!50!black, opacity=0.8] table [x expr=\coordindex, y=\futcol]{\datatable};
				\addlegendentry{FUT};
				
				\addplot [name path = mp, mark=x, purple!60!black, opacity=0.8 ] table [x expr=\coordindex, y=\mpcol]{\datatable};
				\addlegendentry{MP};
			\end{axis}
		\end{tikzpicture}
		\caption{Average incentive ratio of our rules in dependence of the number of voters in the euclidean model.}
		\label{fig:experimenst2}
	\end{figure}

	\section{Conclusion}
	
	In this paper, we analyze whether there are approval-based distribution rules that are fair, efficient, and approximately strategyproof. To this end, we study the incentive ratio of such rules, which quantifies the maximal multiplicative utility gain of a voter when manipulating. While it turns out that several known rules have a large incentive ratio, we prove that the Nash product rule (\nash) has an incentive ratio of $2$. Since this rule is efficient and satisfies the strongest fairness conditions in the literature, this answers our research question in the affirmative. We further show that the incentive ratio of \nash is optimal subject to some of its desirable properties. Lastly, we conduct an experimental analysis to understand the manipulability of our rules for realistic profiles. 
	
	Our paper points to several appealing follow-up questions. Firstly, we find it interesting to examine the incentive ratio also for other public good settings where strategyproofness is often challenging to satisfy. Further, for approval-based budget division, it may be desirable to analyze the tradeoff between strategyproofness, efficiency, and fairness in more depth by simultaneously approximating some of these properties.

\section*{Acknowledgements}
We thank Florian Brandl for pointing out the incentive ratio to us as well as the anonymous IJCAI reviewers for their helpful feedback. 
This work was mainly conducted while Patrick Lederer and Jeremy Vollen were at UNSW Sydney. During that time, all of the authors were supported by the NSF-CSIRO project on ``Fair Sequential Collective Decision Making".

\clearpage
	
	\appendix
	
	\section{Omitted Proofs from \Cref{subsec:notNash}}\label{app:notNash}
	
	In this appendix, we prove Claims (2) and (3) of \Cref{thm:nonash}.
	
	\nonash*
	\begin{proof} Since Claim (1) has been shown in the main body, we focus on the results regarding \fut and \egal. We discuss the corresponding proofs separately.\medskip
		
		\textbf{\fut}: For the upper bound, we note that \fut is known to satisfy GFS \citep{BBPS21a}, so it holds for every profile~$A$ and voter $i$ that $u_i(\text{\fut}(A))\geq \frac{1}{n}$. Since the maximal utility of every voter is $1$, it immediately follows that the incentive ratio of \fut is at most~$n$.  
		
		For the lower bound, we choose an integer $k\geq 6$ and consider the following two profiles $\mathcal{A}$ and $\mathcal{A}'$ with four candidates $C=\{a,b,c,d\}$ and $2k+7$ voters. We note that these profiles only differ in the approval ballot of the left-most voter $i$. \smallskip
		
		\noindent\begin{tabular}{llllllll}
			${A}$: & $1$: $\{a\}$ & $k-2$: $\{a,b\}$ & $2$: $\{b\}$ & $2$: $\{b,d\}$ \\
			& $3$: $\{a,c\}$ & $k$: $\{c,d\}$ & $1$: $\{d\}$\smallskip\\
			${A}'$: & $1$: $\{a,d\}$ & $k-2$: $\{a,b\}$ & $2$: $\{b\}$ & $2$: $\{b,d\}$ \\
			& $3$: $\{a,c\}$ & $k$: $\{c,d\}$ & $1$: $\{d\}$
		\end{tabular}
		
		We claim that FUT selects the distribution $p$ with $p(a)=\frac{1}{2k+7}$, $p(b)=\frac{k}{2k+7}$, and $p(c)=p(d)=\frac{6+k}{4k+14}$ for $A$. To see this, we note that both $c$ and $d$ maximize the approval score with ${k+3}$ approvals each. Moreover, there are $k$ voters who approve both $c$ and $d$ and thus split their shares between these candidates, so $p(c)=p(d)=\frac{6+k}{2\cdot(2k+7)}$. After this step, $k$ voters approving $b$ and $k-1$ voters approving $a$ remain active. Since both candidates has the same approval scores in~$A$, candidate $b$ will be assigned its share next. In particular, all voters reporting $\{a,b\}$ send their $\frac{1}{n}$ budget to $b$, so $p(b)=\frac{k}{2k+7}$. Finally, $a$ only gets the share of the single voter who reports $\{a\}$, so $p(a)=\frac{1}{2k+7}$. 
		
		On the other hand, FUT chooses the distribution $q$ with $q(a)=\frac{k+1}{2k+7}$, $q(b)=\frac{2}{2k+7}$, $q(c)=0$, and $q(d)=\frac{k+4}{2k+7}$ for~${A}'$. To see this, we observe that $d$ has now the uniquely maximal approval score of $k+4$. Hence, all voters approving $d$ send their share to this candidate and $q(d)=\frac{k+4}{2k+7}$. This leaves $k+1$ active voters approving $a$, $k$ active voters approving $b$, and $3$ active voters approving $c$. Consequently, $a$ is chosen next when $\lambda=\frac{k+3}{k+1}$ because $(k+1)\lambda +1 = k+4$. We also note that, for this $\lambda$, it holds that $k\lambda +2 =k\frac{k+3}{k+1}+2=k+2+\frac{2k}{k+1}<k+4$ and $3\lambda + k=3\frac{k+3}{k+1}+k<
		3\cdot\frac{4}{3}+k=k+4$ as $k\geq 6$. These inequalities show that $a$ is indeed assigned its share before $b$ and $c$, so $q(a)=\frac{k+1}{2k+7}$. Finally, as no active voter approves $c$ after this step, the remaining budget goes to~$b$ and $q(b)=\frac{2}{2k+7}$. 
		
		Lastly, we consider the utility of the left-most voter $i$ in $\mathcal A$ and $\mathcal A'$. Assuming that $A$ is the truthful profile, we have that $u_i(p)=\frac{1}{2k+7}$ and $u_i(q)=\frac{k+1}{2k+7}$, so the incentive ratio $\gamma$ of \fut is lower bounded by $\frac{u_i(q)}{u_i(p)}=k+1=\frac{n-5}{2}$. This completes the proof that $\gamma(\text{\fut})\in\Theta(n)$.\medskip
		
		\textbf{\egal}: 
		For the upper bound on the incentive ratio of \egal, we note that, under this rule, every voter has a utility of at least $\frac{1}{m}$. In more detail, it holds for the uniform distribution over all candidates $p_\textit{uni}$ that $u_i(p_\textit{uni})\geq \frac{1}{m}$ for all $i\in N$. Since \egal maximizes the minimal utility of a voter, this means that $u_i(p)\geq \frac{1}{m}$ for the distribution $p$ chosen by \egal in every profile. Thus, the incentive ratio of \egal is at most $m$ since the maximal utility of a voter is $1$. 
		
		For the lower bound, we choose an integer $k\geq 3$ and consider the following profile $A$ with $k^2$ voters and $2k+2$ candidates $C=\{x_1,\dots, x_k, y_1,\dots, y_{k+2}\}$. For notational convenience, we define the set of voters as $N=[k]^2$, i.e., every voter is denoted by at tuple $(i,j)$, where $i$ is best interpreted as row index and $j$ as column index. 
		\begin{itemize}[topsep=0pt, itemsep=0pt]
			\item Voter $(1,1)$ reports $\{x_1,y_1,\dots, y_{k+2}\}$.
			\item For all $i\in [k]$, $j\in [k-1]$ with $(i,j)\neq (1,1)$, voter $(i,j)$ reports $\{x_i,y_j\}$.
			\item For all $i\in [k-2]$, the voters $(i,k)$ report $\{x_i, y_k, y_{k+1}\}$.
			\item Voter $(k-1,k)$ reports $\{x_{k-1}, y_{k+1}, y_{k+2}\}$ and voter $(k,k)$ reports $\{x_k, y_{k}, y_{k+2}\}$.
		\end{itemize}
		Further, the profile $A'$ is derived from $A$ by changing the ballot of voter $(1,1)$ to $\{y_1\}$. We subsequently display the profiles $A$ and $A'$ for $k=4$. 
		\begin{center}
			\tabcolsep=4pt
			\begin{tabular}{@{}l|cccc@{}}
				$A$: & $j=1$ & $j=2$ & $j=3$ & $j=4$\\\hline
				$i=1$ & $\{x_1,y_1,\dots, y_6\}$ & $\{x_1,y_2\}$ & $\{x_1,y_3\}$ & $\{x_1,y_4,y_5\}$\\
				$i=2$ & $\{x_2,y_1\}$ & $\{x_2,y_2\}$ & $\{x_2,y_3\}$ & $\{x_2,y_4,y_5\}$ \\
				$i=3$ & $\{x_3,y_1\}$ & $\{x_3,y_2\}$ & $\{x_3,y_3\}$ & $\{x_3,y_5,y_6\}$ \\
				$i=4$ & $\{x_4,y_1\}$ & $\{x_4,y_2\}$ & $\{x_4,y_3\}$ & $\{x_4,y_4,y_6\}$ \\
			\end{tabular}\medskip
			
			\begin{tabular}{@{}l|cccc@{}}
				$A'$: & $j=1$ & $j=2$ & $j=3$ & $j=4$\\\hline
				$i=1$ & $\{y_1\}$ & $\{x_1,y_2\}$ & $\{x_1,y_3\}$ & $\{x_1,y_4,y_5\}$\\
				$i=2$ & $\{x_2,y_1\}$ & $\{x_2,y_2\}$ & $\{x_2,y_3\}$ & $\{x_2,y_4,y_5\}$ \\
				$i=3$ & $\{x_3,y_1\}$ & $\{x_3,y_2\}$ & $\{x_3,y_3\}$ & $\{x_3,y_5,y_6\}$ \\
				$i=4$ & $\{x_4,y_1\}$ & $\{x_4,y_2\}$ & $\{x_4,y_3\}$ & $\{x_4,y_4,y_6\}$ \\
			\end{tabular}
		\end{center}
		
		We claim that \egal chooses for $A$ the distribution $p$ with $p(x_i)=\frac{1}{k}$ for all $i\in [k]$. To see this, we observe that $p$ guarantees each voter a utility of $\frac{1}{k}$. 
		We will next show that $p$ is indeed the only distribution that satisfies this. To this end, we let $p'$ denote a distribution such that $u_{i,j}(p')\geq \frac{1}{k}$ for all $(i,j)\in N$ and show that $p'=p$. Consider the following three sets of voters for this: $S_1=\{(i,k+1-i)\colon i\in [k]\}$ are the voters on the counterdiagonal, $S_2=\{(i,k-i)\colon i\in [k-1]\}\cup \{(k,k)\}$ are the voters that are one ''above'' the counterdiagonal and $(k,k)$, and $S_3= \{(i,k-i)\colon i\in [k-2]\}\cup \{(k,1),(k-1,k)\}$ is derived from $S_2$ by replacing the voters $(k-1,1)$ and $(k,k)$ with $(k,1)$ and $(k-1,k)$. Since each set contains $k$ voters, it holds that $\sum_{(i,j)\in S_\ell} u_{i,j}(p')\geq k \cdot \frac{1}{k}=1$ for all $\ell\in \{1,2,3\}$. Next, each set contains for each row $i$ and each column $j$ exactly one voter, and none of them contains voter $(1,1)$. By the construction of $A$, this means for each $\ell\in \{1,2,3\}$ that the approval ballots of the voters in $S_\ell$ are pairwise disjoint. We thus compute
		{\allowdisplaybreaks   
			\begin{align*}
				&1\leq \sum_{(i,j)\in S_1} u_{i,j}(p')=\sum_{i\in [k]} p'(x_i) + \sum_{j\in [k+2]\setminus \{k+2\}} p'(y_j),\\
				&1\leq\sum_{(i,j)\in S_2} u_{i,j}(p')=\sum_{i\in [k]} p'(x_i) + \sum_{j\in [k+2]\setminus \{k+1\}} p'(y_j),\\
				&1\leq\sum_{(i,j)\in S_3} u_{i,j}(p')=\sum_{i\in [k]} p'(x_i) + \sum_{j\in [k+2]\setminus \{k\}} p'(y_j).
			\end{align*}
		}
		Since $p'$ is a distribution, the shares of all candidates sum up to $1$, so we infer form these inequalities that $p'(y_{k+2})=p'(y_{k+1})=p'(y_k)=0$. This means for the voters in the $k$-th column that $u_{i,k}(p')=p'(x_i)$. Because $u_{i,k}(p')\geq \frac{1}{k}$ for all $i\in [k]$, we hence derive that $p'(x_i)=\frac{1}{k}$ for all $i\in [k]$ and $p'=p$. 
		
		Next, for the profile $A'$, \egal chooses the distribution $q$ with $q(y_j)=\frac{2}{2k+1}$ for all $j\in [k-1]$ and $q(y_j)=\frac{1}{2k+1}$ for $j\in \{k,k+1,k+2\}$. We first observe that this distribution guarantees each voter in $A'$ a utility of $\frac{2}{2k+1}$. We will again show that $q$ is the only distribution that meets this condition. Thus, we fix a distribution $q'$ such that $u_{i,j}(q')\geq\frac{2}{2k+1}$ for all $(i,j)$ in $A'$ and show that $q'=q$. Similar to the last case, we consider two sets of voters: $S_1=\{(i,i)\colon i\in [k]\}$ are the voters on the main diagonal and $S_2=\{(i,i)\in [k-2]\}\cup \{(k-1,k), (k,k-1)\}$ replaces the voters $(k-1,k-1)$ and $(k,k)$ with $(k-1,k)$ and $(k,k-1)$. Now, we have again $k$ voters per set, so $\sum_{(i,j)\in S_\ell} u_{i,j}(q')\geq k\cdot \frac{2}{2k+1}$ for $\ell \in \{1,2\}$. Moreover, both sets contain exactly one voter for each row and each column and both of them contain voter $(1,1)$. We hence compute that 
		\begin{align*}
			&\frac{2k}{2k+1}\leq \sum_{(i,j)\in S_1} u_{i,j}(q')=\sum_{i=2}^k q'(x_i) + \!\!\!\!\sum_{j\in [k+2]\setminus \{k+1\}}  \!\!\!\!q'(y_j),\\
			&\frac{2k}{2k+1}\leq\sum_{(i,j)\in S_2} u_{i,j}(q')=\sum_{i=2}^k q'(x_i) + \!\!\!\!\sum_{j\in [k+2]\setminus \{k\}} \!\!\! q'(y_j).
		\end{align*}
		
		Since $q'$ is a distribution, we infer from the first inequality that $q'(x_1)+q'(y_{k+1})\leq \frac{1}{2k+1}$ and from the second one that $q'(x_1)+q'(y_{k})\leq \frac{1}{2k+1}$. On the other hand, it holds that $u_{1,k}(q')=q'(x_1)+q'(y_k)+q'(y_{k+1})\geq \frac{2}{2k+1}$. These three inequalities (and the condition that $q'(x_1)\geq 0$) are only satisfied if $q'(x_1)=0$ and $q'(y_k)=q'(y_{k+1})=\frac{1}{2k+1}$. Further, since $q'(x_1)=0$, it follows for all $j\in [k-1]$ that $u_{1,j}(q')=q'(y_j)$. Hence, we conclude that $q'(y_j)\geq \frac{2}{2k+1}$ for all $j\in [k-1]$. Lastly, the unallocated budget so far is at most $1-\sum_{j\in [k+1]} q'(y_j)\leq \frac{1}{2k+1}$, which needs to be used so that the voters $(k-1,k)$ and $(k,k)$ meet the utility threshold of $\frac{2}{2k+1}$. In the current allocation, both voters obtain a utility of precisely $\frac{1}{2k+1}$ from $y_{k}$ and $y_{k+1}$, respectively. Hence, we need to allocate the remaining budget to $y_{k+2}$ as this is the only commonly approved candidate for both voters. Thus, we conclude that $q'=q$, as desired. 
		
		As the last point, we compute the utility gain for voter $(1,1)$ when deviating from $A$ to $A'$. To this end, we note that utility of this voter in $A$ is $u_{1,1}(p)=p(x_1)=\frac{1}{k}$, whereas her utility for $q$ is $u_{1,1}(q)=\sum_{j\in [k+2]} q(y_j)=1$. Hence, the incentive ratio $\gamma$ of \egal is lower bounded by $\frac{u_{1,1}(q)}{u_{1,1}(p)}=k=\frac{m-2}{2}$. This proves that $\gamma(\text{\egal})\in\Theta(m)$.
	\end{proof}
	
	\section{Omitted Proofs from \Cref{subsec:Nash}}\label{app:Nash}
	
	We now turn to the proof of \Cref{lem:nashaux}.
	
	\nashaux*
	\begin{proof} Fix a concave utility profile $U$ and let $p$ denote the distribution chosen by \nash for $U$. We will prove the three claims of this lemma separately.\medskip
		
		\textbf{Proof of Claim (1):} As our first claim, we will show that $u_i(p)>0$ for all voters $i\in N$. To this end, we observe that, since $u_i$ is non-negative and not always $0$, there is for each voter $i$ a distribution $q_i$ such that $u_i(q_i)>0$. We will show that this implies that $u_i(p_\mathit{uni})>0$ for all $i\in N$, where $p_\mathit{uni}$ is the distribution that assigns a share of $\frac{1}{m}$ to each candidate. 
		To see this, let $p_x$ be the distribution assigning probability $1$ to $x$, for each candidate $x\in C$. Then, it holds that 
		\[p_\mathit{uni}=\frac{1}{m}q_i + \sum_{x\in C} \frac{1-q_i(x)}{m} p_x.\]
		By the concavity and non-negativity of $u_i$, we derive that $u_i(p_\mathit{uni})\geq \frac{1}{m}u_i(q_i) + \sum_{x\in C} \frac{1-q_i(x)}{m} u_i(p_x)\geq \frac{1}{m}u_i(q_i) >0$. Since this holds for every voter $i\in N$, this means that $\prod_{i\in N} u_i(p_\mathit{uni})>0$. Finally, because $p$ maximizes the product of the voters' utilities, $\prod_{i\in N} u_i(p)\geq\prod_{i\in N} u_i(p_\mathit{uni})>0$ and thus also $u_i(p)>0$ for all $i\in N$. \medskip
		
		\textbf{Proof of Claim (2):} For proving this claim, we let $q$ denote an arbitrary distribution; our goal is to show that $\sum_{i\in N}\frac{u_i(q)-u_i(p)}{u_i(p)}\leq 0$. To this end, we first observe that, by Claim (1), $u_i(p)>0$ for all voters $i\in N$, so this sum is well-defined. Next, we consider the line segment $x_t=(1-t)p+tq$ for $t\in [0,1]$, which connects $p$ and $q$ in the simplex. Further, for every voter $i\in N$, we define by $\phi_i(t)=u_i(x_t)$ her utility for $x_t$ and by $s_i(t)=\frac{\phi_i(t)-\phi(0)}{t}$ for $t\in (0,1]$ the slope in her utility. Now, we fix a voter $i\in N$ and aim to prove that the right limit of $s_i$ at $0$ exists and is finite, i.e., $\lim_{t\downarrow 0} s_i(t)\in\mathbb{R}$. To this end, we observe that $\phi_i$ is concave as it is merely a different representation of the voter's utility along the line segment between $p$ and $q$. 
		Consequently, $s_i$ is non-increasing, because it holds for all $t',t\in (0,1]$ with $t'<t$ that 
		{\allowdisplaybreaks
			\begin{align*}
				s_i(t')&=\frac{\phi_i(t')-\phi_i(0)}{t'}\\
				&=\frac{\phi_i\left(\frac{t'}{t} t + (1-\frac{t'}{t})0\right)-\phi_i(0)}{t'}\\
				&\geq \frac{\frac{t'}{t} \phi_i(t)+(1-\frac{t'}{t})\phi_i(0) - \phi_i(0)}{t'}\\
				&=\frac{\phi_i(t)-\phi_i(0)}{t}\\
				&=s_i(t).
			\end{align*}
		}
		
		Next, we show that $s_i$ is also bounded. Assume for contradiction that this is not true. Since $s_i$ is bounded from below by $s_i(1)$ (as $s_i$ is non-increasing), there is for every value $r\in\mathbb{R}$ a point $t\in (0,1]$ such that $s_i(t)>r$. Moreover, since $s_i$ is non-increasing, we may assume that $t$ is arbitrarily close to $0$. Our goal is to show that this contradicts that $p$ maximizes the Nash welfare. To this end, let $M$ be a constant such that $\phi_i(t)\leq M$ for all $t\in [0,1]$. 
		Such a constant exists since $\phi_i$ is a concave, non-negative, and real-valued function. In particular, concavity shows for every $\delta\in [0,\frac{1}{2}]$ that $\phi_i(\frac{1}{2})\geq \frac{1}{2}\phi_i(\frac{1}{2}-\delta)+\frac{1}{2}\phi_i(\frac{1}{2}+\delta)$. 
		By non-negativity, this means that $2\phi_i(\frac{1}{2})\geq \phi_i(\frac{1}{2}-\delta)$ and $2\phi_i(\frac{1}{2})\geq \phi_i(\frac{1}{2}+\delta)$ for $\delta \in [0,\frac{1}{2}]$ and thus $2\phi_i(\frac{1}{2})\geq \phi_i(x)$ for all $x\in [0,1]$. 
		
		We will choose a value $t^*\in (0,\frac{1}{2}]$ such that $s_i(t^*)>2nM$ and show that the distribution $x_{t^*}$ has a higher Nash welfare than $x_0=p$. To this end, we observe that, by concavity and non-negativity of the individual utility functions, it holds that $u_j(x_{t^*})\geq (1-t^*)u_j(x_0)>0$ for all $j\in N$. 
		Hence, it suffices to prove that $\sum_{j\in N} \log u_j(x_{t^*})>\sum_{j\in N} \log u_j(x_0)$, or, equivalently, that $\frac{1}{t^*}\sum_{j\in N} (\log u_j(x_{t^*})- \log u_j(x_0))>0$. We first consider the voters $j\in N\setminus \{i\}$. For each of these voters, it holds that $ \log u_j(x_{t^*})- \log u_j(x_0)\geq  \log (1-t^*)u_j(x_0)- \log u_j(x_0)=\log \frac{(1-t^*) u_j(x_0)}{u_j(x_0)}=\log (1-t^*)$. Further, it holds for every $t\in (0,1)$ that $\log 1-t\geq -\frac{t}{1-t}$. Since $t^*\in (0,\frac{1}{2}]$, it follows that $\log1-t^*\geq- \frac{t^*}{1-t^*}\geq -2t^*$. In summary, this means that 
		\begin{align*}
			\sum_{j\in N\setminus \{i\}} \frac{\log u_j(x_{t^*}) -\log u_j(x_0)}{t^*}&\geq \sum_{j\in N\setminus \{i\}}-\frac{t^*}{t^*(1-t^*)}\\
			&\geq -2(n-1).
		\end{align*}
		
		We next turn to voter $i$ and note that $\phi_i(t^*)>\phi_i(0)$ since $s_i(t^*)=\frac{\phi_i(t^*)-\phi_i(0)}{t^*}>0$. Further, by the mean value theorem applied to the logarithm, there is for all values $a,b\in\mathbb{R}_{>0}$ with $a<b$ a constant $c\in (a,b)$ such that $\log b - \log a=\frac{b-a}{c}$. Setting $b=\phi_i(t^*)$ and $a=\phi(0)$, this shows that there is a constant $c\in (\phi_i(0),\phi_i(t^*))$ such that $\log\phi_i(t^*)-\log\phi_i(0)=\frac{\phi_i(t^*)-\phi_i(0)}{c}$. Because $c\leq \phi_i(t^*)\leq M$, we have that $ \log\phi_i(t^*)-\log\phi_i(0)\geq \frac{\phi_i(t^*)-\phi_i(0)}{M}$. By dividing by $t^*$ and using the definition of $s_i$, we  get that 
		\[ \frac{\log \phi_i(t^*) -\log \phi_i(0)}{t^*}\geq \frac{\phi_i(t^*)-\phi_i(0)}{t^*M}=\frac{s_i(t^*)}{M}> 2n.\]	
		
		In combination, this proves that $\frac{1}{t^*}\sum_{j\in N} \log u_j(x_{t^*})- \log u_j(x_0)>0$, as desired. This contradicts that $x_0=p$ maximizes the Nash welfare, so there is a constant $M'$ such that $s_i(t)\leq M'$ for all $t\in (0,1]$. Finally, since $s_i$ is non-increasing and bounded,  $\lim_{t\downarrow 0} s_i(t)$ exists and is finite. 
		
		Based on this observation, we  define $z_i=\lim_{t\downarrow 0} s_i(t)$ for all $i\in N$ and note that this value has multiple important implications. Firstly, since all $s_i$ are non-increasing, it holds that $s_i(t)\leq  z_i$ for all $t\in (0,1]$ and thus $\phi_i(t)-\phi(0)\leq t z_i$. This shows that $\phi_i$ is (right) continuous at $0$. Further, by setting $t=1$, we get that 
		$u_i(q)-u_i(p)=u_i(x_1)-u_i(x_0)\leq z_i$.
		
		Now, to derive our separation inequality, we consider the Nash welfare along our line segment and let $\Phi(t)=\sum_{i\in N} \log \phi_i(t)$ for all $t\in [0,1)$. This function is well-defined since concavity and non-negativity imply that $\phi_i(t)\geq (1-t) u_i(x_0)>0$. Further, since $p=x_0$ maximizes the Nash welfare, we have that $\Phi(0)\geq \Phi(t)$ for all $t\in (0,1)$, or, equivalently, that $\frac{\Phi(t)-\Phi(0)}{t}\leq 0$ for all $t\in (0,1)$. This means that $\lim\sup_{t\downarrow 0} \frac{\Phi(t)-\Phi(0)}{t}\leq 0$, too. Motivated by this observation, we will consider the term $\lim\sup_{t\downarrow 0} \frac{\log \phi_i(t)-\log\phi_i(0)}{t}$ for the individual voters $i\in N$. Using again the mean value theorem for the logarithm, there is for every $t\in (0,1)$ a constant $c(t)$ between $\phi_i(t)$ and $\phi_i(0)$ such that $\log \phi_i(t)-\log\phi_i(0)=\frac{\phi_i(t)-\phi_i(0)}{c(t)}$. Further, by right continuity of $\phi_i$ at $0$, it holds that $\phi_i(t)$ converges to $\phi_i(0)$ as $t$ approaches $0$. This implies that $c(t)$ also converges to $\phi_i(0)$ because it lies always between $\phi_i(0)$ and $\phi_i(t)$. Using the definition of $s_i$, we conclude that 
		\begin{align*}
			\lim\sup_{t\downarrow 0} \frac{\log \phi_i(t)-\log\phi_i(0)}{t}
			&=\lim\sup_{t\downarrow 0} \frac{ \phi_i(t)-\phi_i(0)}{t c(t)}\\
			&=\lim\sup_{t\downarrow 0} \frac{ \phi_i(t)-\phi_i(0)}{t \phi(0)} \\
			&= \frac{z_i}{\phi(0)}. 
		\end{align*}
		
		Since this holds for all voters $i\in N$, it follows that 
		\begin{align*}
			\lim\sup_{t\downarrow 0} \frac{\Phi(t)-\Phi(0)}{t}
			&=\lim\sup_{t\downarrow 0} \sum_{i\in N} \frac{\log \phi_i(t)-\log \phi_i(0)}{t}\\
			&=\sum_{i\in N} \frac{z_i}{\phi_i(0)}.
		\end{align*}
		
		Finally, by combining our insights, we derive that 
		\begin{align*}
			\sum_{i\in N} \frac{u_i(q)-u_i(p)}{u_i(p)}
			&=\sum_{i\in N} \frac{\phi_i(1)-\phi_i(0)}{\phi_i(0)}\\
			&\leq \sum_{i\in N} \frac{z_i}{\phi_i(0)}\\
			&=\lim\sup_{t\downarrow 0} \frac{\Phi(t)-\Phi(0)}{t}\\
			&\leq 0.
		\end{align*}	
		This concludes the proof of our claim.\medskip
		
		\textbf{Proof of Claim (3):} For our last claim, we analyze the change in the Nash welfare when a voter $i$ abstains from the election. To this end, let $q=\text{\nash}(U_{-i})$ be the chosen distribution when voter $i$ abstains. Our goal is to show that $e^{-1}\leq \frac{\prod_{j\in N\setminus \{i\}} u_j(p)}{\prod_{j\in N\setminus \{i\}} u_j(q)}\leq 1$. First, we note that $\prod_{j\in N\setminus \{i\}} u_j(q)>0$ and $\prod_{j\in N\setminus \{i\}} u_j(p)>0$ by Claim (1). For our upper bound, it now suffices to observe that, by definition, $q$ maximizes $\prod_{j\in N\setminus \{i\}} u_j(q')$ among all distributions $q'$. Thus, we immediately get that 
		$\frac{\prod_{j\in N\setminus \{i\}} u_j(p)}{\prod_{j\in N\setminus \{i\}}  u_j(q)}\leq 1$.
		
		For the lower bound, we use the separation inequality in Claim (2). By applying this inequality for $q$, we derive that $\sum_{j\in N} \frac{u_j(q)-u_j(p)}{u_j(p)}\leq 0$. Next, we define  for every voter $j\in N$ a constant $d_j$ such that $1+d_j=\frac{u_j(q)}{u_j(p)}$. By the definition of these constants, it holds that $\frac{u_j(q)-u_j(p)}{u_j(p)}=\frac{(1+d_j)u_j(p)-u_j(p)}{u_j(p)}=d_j$ for all $j\in N$. By substituting this into our separation inequality for all voters $j\in N\setminus \{i\}$ and rearranging it, we derive that $\sum_{j\in N\setminus \{i\}} d_j\leq \frac{u_i(p)-u_i(q)}{u_i(p)}$.
		Next, since our utilities are non-negative, the right hand side is lower bounded by $\frac{u_i(p)-u_i(q)}{u_i(p)}\leq \frac{u_i(p)}{u_i(p)}=1$, which implies that $\sum_{j\in N\setminus \{i\}} d_j\leq1$. 
		
		For the next step, we recall Lemma 3.3 of \citet{CGG13a}: for any finite set of numbers $X=\{\delta_1,\dots, \delta_t\}$ such that $\sum_{j\in X} \delta_j\leq b$ for some constant $b$, it holds that $\prod_{j\in X} (1+\delta_i)\leq (1+\frac{b}{t})^t$. Applying this lemma with $X=\{d_j\colon j\in N\setminus\{i\}\}$ and $b=1$ shows that 
		\[\prod_{j\in N\setminus \{i\}} (1+d_j)\leq \left(1+\frac{1}{n-1}\right)^{n-1}.\]\vspace{-0.1cm}
		
		Now, since $\log \left(1+\frac{1}{n-1}\right)^{n-1} = (n-1) \log (1+\frac{1}{n-1})\leq (n-1)/(n-1)=1$, it holds that $\left(1+\frac{1}{n-1}\right)^{n-1}\leq e$. This proves our lower bound because
		\[\frac{\prod_{j\in N\setminus \{i\}} u_j(p)}{\prod_{j\in N\setminus \{i\}} u_j(q)}=\frac{1}{\prod_{j\in N\setminus \{i\}}(1+d_j)}\geq \frac{1}{e}.\]
		This completes the proof of this lemma. 
	\end{proof}

	\section{Detailed Discussion of Experiments}\label{app:experiments}
	
	In this appendix, we further discuss the setup of our experiments and present additional statistics in \Cref{fig:IC,fig:EUC,fig:MALL}. We note that our experiments rely on the PrefSampling package by \citet{BFL+24a}.
	
	First, we note that we focus on relatively small numbers of voters and candidates. This is necessary due to computational issues: to compute the value $IR(f,A)$ for a rule $f$ and a profile $A$, we compute $f$ for the original profile $A$ and each possible deviation of each voter. For example, if $n=100$ and $m=10$, we need to compute $f$ (in the worst-case) $(2^{10}-2)\cdot 100\approx 100,000$ times. This is computationally challenging since, e.g., \egal requires us to repeatedly solve linear programs. Also, in particular the exponential blow-up for $m$ prevents us from running computer experiments for more candidates. A second reason why we focus on small values of $n$ is that, for large $n$, the effect of a single voter becomes more and more negligible and it seems more plausible to reason about group manipulations. Lastly, by using small numbers of voters, we also avoid numerical issues that may arise otherwise as the utilities of voters can become very small if $m$ and $n$ are large. That said, we find more extensive experiments desirable and leave this for future work. 
	
	Next, we discuss additional insights form our experiments for each distribution. 
	
	\paragraph{Impartial culture model.} First, in our impartial culture model, each voter approves each candidate with probability $0.3$. This means that, in expectation, each candidate is approved by the same number of voters (namely $0.3n$) and that there is no correlation between candidates. 
	
	This setup has clear consequences for our rules. Firstly, for \nash and \egal, this means that, as $n$ grows, the corresponding distributions become closer to the uniform distribution. As a consequence, the average incentive ratio of \nash slightly decreases as $n$ increases. For \egal, this effect is not visible in our experiments. We believe the reason for this is that \egal does not converge as smoothly as \nash to the uniform lottery, so a larger number of voters may be necessary. However, for large enough $n$, \egal should return the uniform lottery over all candidates for all our distributions and the percentage of manipulable profiles should decline (because each singleton ballot will be reported by some voter and \egal is not manipulable in such profiles). 
	
	By contrast, \fut and especially \map are very manipulable under the impartial culture model. The reason for this is that these rules assign the shares sequentially to the candidates, and the sequence of candidates is roughly guided by the approval scores. Hence, if all candidates have similar (or even the same) approval scores, it becomes easy to change the sequence in which the candidates are processed, which is necessary for a successful manipulation. Moreover, for \map, there is a clear trend that the manipulability increases as $n$ increases, whereas the average incentive ratio of \fut is relatively constant.
	
	\paragraph{Euclidean Model.}
	In our Euclidean model, we place candidates and voters uniformly at random in the $3$-dimensional unit ball $\{x\in \mathbb{R}^3\colon |x|_2\leq 1\}$. Further, a voter approves each candidate with a distance of at most $0.4$ from her position. We note that this radius is empirically chosen so that voters approve roughly three to four candidates on average. While the approval scores of the candidates are, in expectation, still roughly the same, this model introduces correlation between candidates. Specifically, candidates that are close to each other in the unit ball will be jointly approved together by many voters, whereas candidates that are far apart from each other cannot be jointly approved. 
	
	In general, the incentive ratio of our rules tends to be smaller for this model than for the impartial culture model, which is especially visible or \map. Moreover, we note that \map and \fut are significantly less often manipulable in the Euclidean model than under the impartial culture, which indicates that the additional structure in this model benefits these rules. Further, in this model, the incentive ratio of \map and \egal clearly increases as $n$ increases. For \map, a possible explanation for this is that, the more voters we have, the more likely there is a pair of candidates $x,y$ such that many voters approve both $x$ and $y$ (as they are close to each other in the unit ball) but there are some voters how only approve $x$ and $y$ respectively. \map will then assign a large share to either $x$ or $y$ by the tie-breaking, opening up the possibility for the voters who only approve the other candidate to manipulate by only voting for this candidate. For \egal, we believe that a similar effect happens due to possible free-riding manipulations as the number of voters increases. 
	
	\paragraph{Truncated Mallow's model.} In our last model, we sample first strict preference relation $\succ_i$ according to Mallow's model with a dispersion parameter of $\phi=0.75$ for each voter. Specifically, Mallow's model is defined by a fixed preference relation $\rhd=x_1,\dots,x_m$ over the candidates. The probability to assign a voter a preference relation $\succ_i$ is then proportional to $\phi^{|\{(x,y)\in A^2\colon x\rhd y\cap y\succ_i x\}|}$. As as second step, we sample a threshold value $t_i\in \{1,\dots, m-1\}$ and each voter $i$ approves the $t_i$ favorite alternatives in her preference relation. We empirically chose the dispersion parameter $\phi$ such that the approval score of the most approved candidate is roughly 2.5 times that of the least approved candidate. Further, this model introduces both a correlation between candidates and differences in the approval scores. Specifically, candidates that are close to each other in the ranking $\rhd$ are more likely to be jointly approved and higher-ranked candidates in this ranking obtain more approvals. 
	
	The additional structure in this model helps to further reduce the average incentive ratio, but the number of manipulable profiles remains roughly the same as for the euclidean model. Specifically, for \nash, \fut, and \map, the average and maximum incentive ratio as well as the 90\%-percentile are consistently the lowest observed in our experiments. This is most notable for \map, where the 90\%-percentile of the incentive ratio is now always less or equal to $1.5$. We note that we expected this model to be the least manipulable for \fut and \map: since there is now a clear difference in the approval scores of the candidates, the first rounds of these rules will often be non-manipulable. Since the majority of the probability is assigned in these rounds, we end up with only small manipulations towards the end. Moreover, if we were to use a smaller dispersion parameter, these effects get even more drastic. On the other hand, \egal is agnostic to these effects as approval scores are largely irrelevant for this rule. However, under the truncated Mallow's model, there are frequently candidates that obtain a share of $0$ despite being approved by some voters, which allows for free-riding manipulations.

	\begin{figure*}
		\small
		\tabcolsep=4pt
		\begin{filecontents*}{./simulations/IC.csv}
n;nash_avg;nash_std;nash_max;nash_per90;nash_min;nash_freq;egal_avg;egal_std;egal_max;egal_per90;egal_min;egal_freq;mp_avg;mp_std;mp_max;mp_per90;mp_min;mp_freq;fut_avg;fut_std;fut_max;fut_per90;fut_min;fut_freq
10;1.2043513387397653;0.08982166156979898;1.4542573661681704;1.3193185961261242;1.000002732383368;0.995;1.6426415490772088;0.2101219420328746;2.693877551020404;1.875;1.0000000000000002;0.999;1.3126972222222224;0.35288406238070946;5.0;1.6;1.0;0.785;1.0648322104424097;0.10371810718503655;2.0;1.2000000000000002;1.0;0.455
20;1.2016548310814092;0.06839359928478111;1.479054537964556;1.291068086330517;1.039357783782492;1.0;1.8267697083219914;0.25503910126057333;2.9333333333333345;2.1428571428571432;1.3333333333333337;1.0;1.4863101509386725;0.5938611968213623;6.000000000000001;1.777777777777778;1.0;0.951;1.128614128523579;0.1625268396116528;3.0000000000000004;1.333333333333333;1.0;0.739
30;1.1895557261717804;0.06513729605876019;1.4324568665713515;1.27600572585283;1.0229362851180088;1.0;1.861587127735259;0.2494981134974716;3.085714285714287;2.2000000000000006;1.3333333333333344;1.0;1.6889485899773864;1.0378740547973937;10.000000000000002;2.333333333333333;1.0;0.966;1.159652063488571;0.21932292901071962;4.0;1.3333333333333333;1.0;0.815
40;1.1756803815634882;0.05837471131191782;1.409030917845413;1.2551916385403188;1.0228881059213526;1.0;1.8567395734774659;0.2513655186161294;3.6111111111111103;2.1875;1.3333333333333333;1.0;1.8729259582058388;1.5147248670484181;15.0;2.9999999999999996;1.0000000000000002;0.979;1.1855186645694338;0.21001654095963337;3.0000000000000004;1.4285714285714284;1.0;0.851
50;1.1671996639165731;0.05642935731452605;1.3650360905986054;1.240675201574824;1.0229789070964217;1.0;1.8635440354133488;0.2717089764998427;3.230769230769231;2.25;1.2000000000000008;1.0;2.0266834326806853;2.019621744199553;24.0;3.0;1.0;0.989;1.1874950528561508;0.20071391424654816;3.0;1.4166666666666667;1.0;0.896
60;1.1607459357387468;0.05649825602747604;1.3817112208056126;1.2387678017599908;1.0122937288177687;1.0;1.8365613427682939;0.2577882683480119;3.1428571428571437;2.1875;1.000000000000001;0.998;2.1057656026178213;1.9485020365901577;21.0;4.0;1.0000000000000002;0.993;1.1934841366328424;0.2010721879408125;3.0;1.4;1.0;0.906
70;1.145895757160761;0.06086544721387877;1.3863679734353926;1.2253683511698386;1.0097281074409568;1.0;1.8175469651454905;0.25690581251546996;3.000000000000001;2.1875;1.000000000000001;0.998;2.252484566678032;2.3027869671736565;19.0;4.0;1.0000000000000002;0.984;1.1739326616936852;0.17891387335098477;3.0;1.3333333333333337;1.0;0.885
80;1.1356782337958298;0.0600086555850233;1.3548156123162445;1.2114438021601943;1.0121430723963225;1.0;1.7794320255857852;0.27153990306986536;2.8000000000000003;2.1875;1.0000000000000009;0.99;2.3851852999435277;2.608885803393238;25.0;5.0;1.0000000000000002;0.994;1.1790444552308315;0.16963215148663743;2.5;1.3333333333333337;1.0;0.927
90;1.1282993464783735;0.057535223225278806;1.3133563368495575;1.201836677902163;1.0122031063200043;1.0;1.7595375218984224;0.27473625402435764;2.785714285714287;2.1875;1.0000000000000009;0.983;2.4964150936323577;2.918556890408279;37.0;4.999999999999999;1.0000000000000002;0.994;1.1716385945570182;0.16449110364832048;2.5999999999999996;1.3333333333333337;1.0;0.917
100;1.1216752890595678;0.061279290316656;1.3183974095418285;1.203395989499315;1.0081023412070096;1.0;1.7351233876896177;0.28456743367768833;2.799999999999999;2.1875;1.0000000000000009;0.967;2.522792714913917;2.8209921067861328;26.0;5.333333333333334;1.0000000000000002;0.992;1.1710679698253885;0.1696912014680412;3.0;1.3333333333333337;1.0000000000000002;0.921
\end{filecontents*}
\pgfplotstabletypeset[
		col sep=semicolon,
		columns={n,nash_avg,nash_max,nash_per90,nash_std,nash_freq,
			egal_avg,egal_max,egal_per90,egal_std,egal_freq,
			fut_avg,fut_max,fut_per90,fut_std,fut_freq,
			mp_avg,mp_max,mp_per90,mp_std,mp_freq},
		every column/.style={
			sci=false,           
			fixed,               
			fixed zerofill,      
			precision=2          
		},
		columns/n/.style={fixed,               
			fixed zerofill,      
			precision=0},
		columns/nash_avg/.style={column type=|c},
		columns/egal_avg/.style={column type=|c},   
		columns/fut_avg/.style={column type=|c},  
		columns/mp_avg/.style={column type=|c},
		every head row/.style={
			output empty row,         
			before row={
				\toprule
				& \multicolumn{5}{c|}{\nash} & \multicolumn{5}{c|}{\egal} & \multicolumn{5}{c|}{\fut} & \multicolumn{5}{c}{\map}\\
				n & Avg & Max & 0.9P & Std & Freq & Avg & Max & 0.9P & Std & Freq & Avg & Max & 0.9P & Std & Freq & Avg & Max & 0.9P & Std & Freq \\ 
				\midrule
			},
		},
		every last row/.style={after row=\bottomrule}
		]{./simulations/IC.csv}
		\caption{Full summary of our simulations with the impartial culture model. Avg, Max, 0.9P, and Std respectively indicate the average, maximum, 90\%-percentile, and standard deviation of the incentive ratio. Freq denotes the percentage of manipulable profiles.}
		\label{fig:IC}
	\end{figure*}
	
	\begin{figure*}
		\small
		\tabcolsep=4pt
		\begin{filecontents*}{./simulations/EUC.csv}
n;nash_avg;nash_std;nash_max;nash_per90;nash_min;nash_freq;egal_avg;egal_std;egal_max;egal_per90;egal_min;egal_freq;mp_avg;mp_std;mp_max;mp_per90;mp_min;mp_freq;fut_avg;fut_std;fut_max;fut_per90;fut_min;fut_freq
10;1.1472146448760945;0.11379683149232497;1.5666358515130203;1.303886545710816;1.0000000000000004;0.891;1.4980958994708995;0.15117249668404553;2.25;1.6000000000000003;1.0;0.988;1.2160321428571428;0.27457290485796637;2.9999999999999996;1.5000000000000002;1.0;0.644;1.046684670363563;0.08919616084287597;2.0;1.1666666666666663;1.0;0.34
20;1.1988246676848362;0.10803380058248546;1.5647874956784753;1.345766618753078;1.0000039059550938;0.999;1.623471947496948;0.14686957359855554;2.4000000000000004;1.777777777777778;1.0000000000000004;0.997;1.3562400823634568;0.5662226466367106;6.999999999999999;1.6249999999999998;1.0;0.868;1.1072270640048274;0.15789640399957394;3.0000000000000004;1.25;1.0;0.631
30;1.2018988559284904;0.10274786558337781;1.5671822641208295;1.3431428299700925;1.0022462979311797;1.0;1.6586117629592636;0.14543404911402977;2.666666666666668;1.8000000000000005;1.0000000000000004;0.999;1.512887219241239;1.0080157486910135;10.0;2.0;1.0;0.904;1.122214604889376;0.15453702671347894;2.0;1.3;1.0;0.684
40;1.1932866104190154;0.0972451768863932;1.5860819967746964;1.33909803567279;1.0111003584241005;1.0;1.6910181569819074;0.15163408023482355;2.666666666666667;1.8333333333333337;1.200000000000001;1.0;1.5307000406565936;0.9604426166625744;11.0;2.0;1.0;0.913;1.1368532206408932;0.18928600245789212;3.0000000000000004;1.3333333333333333;1.0;0.714
50;1.1969450680504343;0.09852231458713505;1.573016729998305;1.3314986941397535;1.017709985357458;1.0;1.706718087468088;0.14638523724746202;2.777777777777777;1.833333333333334;1.3124999999999996;1.0;1.665476164751282;1.4024330094479773;18.0;2.5;1.0;0.913;1.1532935794189414;0.20300133308489204;3.0;1.3529411764705883;1.0;0.76
60;1.1874912062435934;0.10216594757176972;1.584501092669836;1.33035477322295;1.008125972697611;1.0;1.7157568986568992;0.15925003897666865;2.7857142857142865;1.875;1.0000000000000004;0.999;1.7256048110681343;1.7497194877528015;18.0;2.5;1.0;0.925;1.1473679025292072;0.22786037665679265;5.0;1.3333333333333333;1.0;0.771
70;1.183713381123412;0.09662299595729025;1.634578839088172;1.303994454272001;1.0109248146285883;1.0;1.723973537500512;0.15092709645769176;2.5000000000000004;1.875000000000001;1.000000000000005;0.999;1.777609235747577;1.9979895069814906;23.0;2.5;1.0;0.909;1.144793841608468;0.19304009648840093;3.0;1.3571428571428572;1.0;0.765
80;1.17964130032321;0.10088529758797332;1.5944339092874154;1.3134623077458842;1.010502308799751;1.0;1.731984770308668;0.16964667998447874;2.571428571428574;1.9285714285714304;1.0000000000000016;0.999;1.7099188903835048;1.70897866640851;22.0;2.666666666666667;1.0;0.898;1.1456997822036934;0.19751123497245876;3.5;1.3636363636363635;1.0;0.733
90;1.176720175310466;0.10061128355839649;1.6373070196351402;1.3056838397542427;1.0117217031121728;1.0;1.735377924390425;0.156569861070031;2.5714285714285716;1.885714285714287;1.3124999999999996;1.0;1.7906332699871337;2.02892909751951;25.999999999999996;2.727272727272727;1.0;0.881;1.1430535190070168;0.2808384335295393;7.0;1.3333333333333337;1.0;0.733
100;1.1750844112570762;0.09882594846183189;1.5957891960410868;1.311441494675345;1.0081564659875286;1.0;1.7362143238705743;0.15944619189429862;2.57142857142857;1.857142857142857;1.3333333333333333;1.0;1.7872557802410396;2.306614495165255;30.0;2.5;1.0;0.891;1.1426279905654428;0.20589368661657276;3.0;1.352941176470588;1.0;0.752
\end{filecontents*}
\pgfplotstabletypeset[
		col sep=semicolon,
		columns={n,nash_avg,nash_max,nash_per90,nash_std,nash_freq,
			egal_avg,egal_max,egal_per90,egal_std,egal_freq,
			fut_avg,fut_max,fut_per90,fut_std,fut_freq,
			mp_avg,mp_max,mp_per90,mp_std,mp_freq},
		every column/.style={
			sci=false,           
			fixed,               
			fixed zerofill,      
			precision=2          
		},
		columns/n/.style={fixed,               
			fixed zerofill,      
			precision=0},
		columns/nash_avg/.style={column type=|c},
		columns/egal_avg/.style={column type=|c},   
		columns/fut_avg/.style={column type=|c},  
		columns/mp_avg/.style={column type=|c},
		every head row/.style={
			output empty row,         
			before row={
				\toprule
				& \multicolumn{5}{c|}{\nash} & \multicolumn{5}{c|}{\egal} & \multicolumn{5}{c|}{\fut} & \multicolumn{5}{c}{\map}\\
				n & Avg & Max & 0.9P & Std & Freq & Avg & Max & 0.9P & Std & Freq & Avg & Max & 0.9P & Std & Freq & Avg & Max & 0.9P & Std & Freq \\ 
				\midrule
			},
		},
		every last row/.style={after row=\bottomrule}
		]{./simulations/EUC.csv}
		\caption{Full summary of our simulations with the euclidean model. Avg, Max, 0.9P, and Std respectively indicate the average, maximum, 90\%-percentile, and standard deviation of the incentive ratio. Freq denotes the percentage of manipulable profiles.}
		\label{fig:EUC}
	\end{figure*}
	
	\begin{figure*}
		\small
		\tabcolsep=4pt
		\begin{filecontents*}{./simulations/MALL.csv}
n;nash_avg;nash_std;nash_max;nash_per90;nash_min;nash_freq;egal_avg;egal_std;egal_max;egal_per90;egal_min;egal_freq;mp_avg;mp_std;mp_max;mp_per90;mp_min;mp_freq;fut_avg;fut_std;fut_max;fut_per90;fut_min;fut_freq
10;1.182434783517996;0.11287037468925419;1.603530618296375;1.3263240314306053;1.0000000932534538;0.888;1.4974730260480265;0.18858781221871376;2.5714285714285716;1.7142857142857149;1.0;0.978;1.1498670634920636;0.18959010422117878;2.9999999999999996;1.3333333333333337;1.0;0.617;1.0323621430187986;0.07697327304382862;2.0;1.1111111111111112;1.0;0.314
20;1.1799293633237995;0.08357301352479345;1.520812363332649;1.2920817405492706;1.0006744977829178;0.999;1.6425119196554254;0.19608005462083253;3.062500000000003;1.8750000000000004;1.2000000000000002;1.0;1.2152144836630518;0.3067718863829331;5.0;1.4;1.0;0.823;1.0649199833203715;0.11396641769811405;2.0;1.1666666666666667;1.0;0.55
30;1.165495868774083;0.07920149855910522;1.4989818280775549;1.266924156818;1.0109299896432105;1.0;1.7129797194952419;0.21022901188065252;3.0000000000000004;2.0000000000000004;1.25;1.0;1.2568341756107757;0.4272954850316999;7.0;1.4999999999999998;1.0;0.883;1.0887366023777607;0.14864031839031822;2.0;1.2000000000000002;1.0;0.637
40;1.1529386302908715;0.08103943772730823;1.4378394709780318;1.2605572964298433;1.010239530542674;1.0;1.72516636819386;0.1993334183647043;2.666666666666667;2.000000000000001;1.2250000000000003;1.0;1.285118222837052;0.5542122978413307;11.0;1.5;1.0;0.89;1.0912451974005326;0.13316573730246578;2.0;1.25;1.0;0.679
50;1.1409531359394938;0.0839547061280852;1.451153845639129;1.258284127180688;1.0073602594956477;1.0;1.7193355019350594;0.17304944675629008;2.571428571428572;2.0;1.000000000000001;0.999;1.3473940700008529;0.8556270145045157;12.0;1.5;1.0;0.911;1.1018333755960872;0.15716553490657287;2.5;1.25;1.0;0.684
60;1.1315642124273106;0.08437793788591486;1.4256601341011887;1.2543290835569008;1.0046921562146216;1.0;1.735389503350248;0.17540740468601637;2.7600000000000007;2.0;1.1250000000000009;1.0;1.277467572959953;0.4922772854034016;7.0;1.5;1.0;0.901;1.093682460267495;0.1279694826276359;2.0;1.25;1.0;0.684
70;1.1228851539414677;0.08152008351421419;1.3936915307283992;1.2420945593536918;1.0068158660074482;1.0;1.7383958923559288;0.1657926739916441;2.7777777777777777;1.928571428571428;1.0000000000000004;0.998;1.2862629267805175;0.4913005068017557;8.0;1.5;1.0;0.912;1.1021770381254743;0.14258202822959926;3.0;1.25;1.0;0.698
80;1.1220446607721215;0.08844068946206123;1.4374226696841186;1.2495705075451686;1.0060121410503249;1.0;1.7284930868287494;0.15835367128589342;2.571428571428571;1.8000000000000005;1.0000000000000009;0.996;1.274895284585569;0.49221016952143476;10.0;1.5;1.0;0.912;1.1010055088007376;0.13725751708858744;2.333333333333333;1.25;1.0;0.699
90;1.1191720083898353;0.09246007534662609;1.4253214655441324;1.2508267687167354;1.0076251629246182;1.0;1.7397954787247072;0.16551484606329694;2.3571428571428608;1.800000000000004;1.0000000000000009;0.997;1.284598525403677;0.5351411909392136;11.0;1.5;1.0;0.934;1.1054259206884516;0.1347991642510282;2.0;1.25;1.0;0.715
100;1.1112759312262765;0.08993324120139183;1.457262418439474;1.239685309068923;1.0067389983850792;1.0;1.7306625840414178;0.17980436959876436;2.400000000000001;1.8000000000000005;1.0000000000000009;0.988;1.290829498307249;0.599753229044262;11.5;1.5;1.0;0.901;1.0962585900349355;0.11722830831083457;2.0;1.25;1.0;0.702
\end{filecontents*}
\pgfplotstabletypeset[
		col sep=semicolon,
		columns={n,nash_avg,nash_max,nash_per90,nash_std,nash_freq,
			egal_avg,egal_max,egal_per90,egal_std,egal_freq,
			fut_avg,fut_max,fut_per90,fut_std,fut_freq,
			mp_avg,mp_max,mp_per90,mp_std,mp_freq},
		every column/.style={
			sci=false,           
			fixed,               
			fixed zerofill,      
			precision=2          
		},
		columns/n/.style={fixed,               
			fixed zerofill,      
			precision=0},
		columns/nash_avg/.style={column type=|c},
		columns/egal_avg/.style={column type=|c},   
		columns/fut_avg/.style={column type=|c},  
		columns/mp_avg/.style={column type=|c},
		every head row/.style={
			output empty row,         
			before row={
				\toprule
				& \multicolumn{5}{c|}{\nash} & \multicolumn{5}{c|}{\egal} & \multicolumn{5}{c|}{\fut} & \multicolumn{5}{c}{\map}\\
				n & Avg & Max & 0.9P & Std & Freq & Avg & Max & 0.9P & Std & Freq & Avg & Max & 0.9P & Std & Freq & Avg & Max & 0.9P & Std & Freq \\ 
				\midrule
			},
		},
		every last row/.style={after row=\bottomrule}
		]{./simulations/MALL.csv}
		\caption{Full summary of our simulations with the truncated Mallow's model. Avg, Max, 0.9P, and Std respectively indicate the average, maximum, 90\%-percentile, and standard deviation of the incentive ratio. Freq denotes the percentage of manipulable profiles.}
		\label{fig:MALL}
	\end{figure*}
	
\end{document}